\documentclass[aps,pra,superscriptaddress,floatfix,notitlepage,nofootinbib,longbibliography,11pt]{revtex4-1} 
\pdfoutput=1
\usepackage{amsmath,amsthm,amsfonts,amssymb,braket}
\usepackage{graphicx}
\usepackage{color}

\makeatletter
\def\l@subsubsection#1#2{}
\makeatother

\newcommand{\ba}{\begin{align}}
\newcommand{\ea}{\end{align}}

\usepackage[colorlinks=true,citecolor=blue,linkcolor=blue]{hyperref}
\usepackage[capitalize,nameinlink]{cleveref}

\newtheorem{lemma}{Lemma}[section]

\theoremstyle{definition}

\newtheorem{remark}[lemma]{Remark}

\newcommand{\Udis}{U_{\rm dis}}
\newcommand{\tUdis}{\tilde U_{\rm dis}}
\newcommand{\tUdisz}{\tilde U_{\rm dis,0}}
\newcommand{\Umat}{U^{\rm mat}}

\newcommand{\RR}{{\mathbb{R}}}
\newcommand{\ZZ}{{\mathbb{Z}}}

\newcommand{\symX}{{\mathbf X}}
\newcommand{\Xbdry}{\symX_{\rm bdry}}
\newcommand{\Xbdryz}{\symX_{\rm bdry,0}}

\begin{document}

\title{An Exactly Solvable Model for a $4+1D$ Beyond-Cohomology Symmetry Protected Topological Phase}
\begin{abstract}
We construct an exactly solvable commuting projector model for a $4+1$ dimensional ${\mathbb Z}_2$ symmetry-protected topological phase (SPT) which is outside the cohomology classification of SPTs.  The model is described by a decorated domain wall construction, with ``three-fermion'' Walker-Wang phases on the domain walls.  We describe the anomalous nature of the phase in several ways.  One interesting feature is that, in contrast to in-cohomology phases, the effective ${\mathbb Z}_2$ symmetry on a $3+1$ dimensional boundary cannot be described by a quantum circuit and instead is a nontrivial quantum cellular automaton (QCA).  A related property is that a codimension-two defect (for example, the termination of a ${\mathbb Z}_2$ domain wall at a trivial boundary) will carry nontrivial chiral central charge $4$ mod $8$. We also construct a gapped symmetric topologically-ordered boundary state for our model, which constitutes an anomalous symmetry enriched topological phase outside of the classification of Ref. \cite{HermeleChen}, and define a corresponding anomaly indicator.

\end{abstract}

\author{Lukasz Fidkowski}
\affiliation{Department of Physics, University of Washington, Seattle WA USA}

\author{Jeongwan Haah}
\affiliation{Microsoft Quantum and Microsoft Research, Redmond, WA USA}

\author{Matthew B.~Hastings}
\affiliation{Station Q, Microsoft Research, Santa Barbara, CA USA}
\affiliation{Microsoft Quantum and Microsoft Research, Redmond, WA USA}

\maketitle

\section{Introduction}

A non-trivial symmetry protected topological (SPT) phase is one that can be continuously connected to a trivial phase, 
but only at the expense of breaking the symmetry or closing the gap.
One well-understood sub-class of SPTs are the so-called ``in-cohomology'' phases, 
whose quantized responses can equivalently be thought of 
either in a Lagrangian field theory formulation as twisted Dijkgraaf-Witten terms~\cite{dijkgraaf1990topological}
or in a Hamiltonian lattice formulation as braiding statistics of symmetry flux defects~\cite{LevinGu, Chen2013}.
Less well-understood are the remaining ``beyond-cohomology'' phases.
At a field theory level these have mixed gauge-gravity terms 
in their response~\cite{Kapustin, Freed, FreedHopkins, debray2018low, Wang_2015}, 
but it is not always clear how to interpret such field theory responses at a lattice Hamiltonian level.  

This issue becomes especially sharp in the case of $4+1$D phases of bosons protected by onsite unitary $\ZZ_2$ symmetry.
Here field theory predicts a $\ZZ_2 \times \ZZ_2$ classification of SPTs, 
generated by an in-cohomology phase and a beyond-cohomology phase~\cite{Kapustin, FreedHopkins}.
The latter corresponds to the term $\frac{1}{2} A \,w_2^2$ in the Lagrangian, 
where $A$ is the $\ZZ_2$ gauge field and $w_2$ the second Stiefel-Whitney class of the spacetime manifold.
Recently a candidate lattice Hamiltonian for the beyond-cohomology phase was proposed~\cite{Freedman_2016, debray2018low}.
This so-called ``Generalized Double Semion'' (GDS) dual model%
\footnote{
The GDS dual is a $\ZZ_2$ SPT~\cite{fidkowski2019disentangling} 
which gauges into the GDS model originally written down in~\cite{Freedman_2016, debray2018low}.
}
does indeed reproduce the correct $\frac{1}{2} A \,w_2^2$ response, 
for spacetime manifolds of the form $M_{\rm{spatial}} \times {\rm{time}}$ and flat $\ZZ_2$ gauge field configurations.
However, in \cite{fidkowski2019disentangling} it was shown that in flat space, 
the GDS dual model is equivalent, up to a finite depth circuit of local $\ZZ_2$-symmetric quantum gates, 
to the in-cohomology SPT Hamiltonian, implying that it cannot be in the beyond-cohomology phase.
This leaves the natural questions:
is there really a $\ZZ_2$-protected beyond cohomology phase in $4+1$D, 
and, if so, what is its quantized response at the Hamiltonian lattice level?

In this paper we answer these questions by constructing an explicit commuting projector model 
for a $\ZZ_2$-protected $4+1$D beyond-cohomology phase,
and identifying a quantized invariant of gapped $\ZZ_2$-symmetric Hamiltonians 
that distinguishes it from the trivial and in-cohomology phases.
The model involves decorating $\ZZ_2$ domain walls with $3+1$D Walker-Wang models based on the ``$3$-fermion'' topological order;
the considerable technical challenges associated with consistently performing this decoration on non-flat $3+1$D geometries 
and fluctuating the domains occupies us for all of \cref{sec:model}.
One useful consequence of our construction is the existence 
of a $\ZZ_2$-symmetric disentangling circuit $\Udis$ for our ground state.%
\footnote{
Although $\Udis$ is overall $\ZZ_2$-symmetric,
the individual gates that make it up cannot all be $\ZZ_2$-symmetric for our model to describe a non-trivial SPT.
}

The $3$-fermion Walker-Wang model is an SPT of time reversal symmetry $\ZZ_2^T$,
so naively one might expect the decorated domain wall model to be an SPT of $\ZZ_2 \times \ZZ_2^T$.
However, it turns out that time reversal symmetry is not necessary,
and the model is an SPT of just the unitary $\ZZ_2$.
To substantiate this claim we define a quantized response invariant 
by probing the bulk with a non-flat $\ZZ_2$ gauge field configuration, 
namely a static $2$-spatial dimensional $\ZZ_2$ symmetry defect.
By choosing the Hamiltonian at the core of the defect appropriately 
we can ensure that there is no topological order (i.e., no anyons) living on the defect.
In a rough sense, which we make precise later, 
the non-trivial signature of the phase is then the fact that 
the defect carries half (modulo one) of the minimal quantized chiral central charge 
allowed for a $2+1$D invertible state of bosons, i.e., its chiral central charge is $4\,\,\rm{mod}\,\,8$.

Equivalently, we can understand the beyond-cohomology phase by studying its $3+1$D boundary, 
where we expect an anomalous action of the $\ZZ_2$ symmetry. 
In \cref{subsec:J} we will show that the signature of this phase 
is also encoded in a special property of this boundary symmetry action $\Xbdry$: 
namely, that $\Xbdry$ is non-trivial as a quantum cellular automaton (QCA)~\cite{Gross_2012, haah2018nontrivial}
--- in particular, it is not a finite depth circuit, despite preserving locality.
Such a boundary action is more severely anomalous than that of an in-cohomology SPT, 
which, despite not being onsite, is still a finite depth circuit.
A non-trivial QCA, on the other hand, cannot even be ``truncated'' to act on a portion of space,
leading to a breakdown at the second step of the Else-Nayak descent procedure 
that characterizes bulk SPT order in terms of boundary symmetry action~\cite{ElseNayak}.

Although the non-trivial QCA nature of $\Xbdry$ constitutes a well-defined quantized invariant, 
it is somewhat abstract, so it is desirable to have a concrete physical diagnostic for when $\Xbdry$ is non-trivial.
This can be done in several ways, 
all of which essentially encode the idea that a $\ZZ_2$ domain wall at the boundary, 
if gapped out in a way that avoids topological order, 
will have a chiral central charge of $4\,\,\rm{mod}\,\,8$.
For example, one can compactify one of the boundary directions and make two domains, 
with two domain walls that run parallel to the remaining two directions; 
the claim, as show in \cref{subsec:l} is then that 
if the domain walls are ``identical,'' in the sense of being related by a translation followed by $\Xbdry$, 
then the $2+1$D dimensionally reduced system has a central charge 
which is an odd multiple of $8$.  
Alternatively, we can study a single domain wall in isolation: 
using the fact that the regions away from the domain wall can be gapped using commuting projectors, 
in \cref{subsec:m} we generalize A.~Kitaev's bulk definition of chiral central charge~\cite[App.D]{Kitaev_2005} 
to allow it to be applied to a single domain wall.

A natural question to ask about anomalous boundaries is: 
what kind of symmetric states can the boundary accommodate?  
For the case of our $4+1$D beyond-cohomology SPT of $\ZZ_2$ 
we will show in \cref{sec:topoboundary} that the boundary anomaly is saturated by a certain $3+1$D $\ZZ_2$ symmetry enriched topological (SET) phase, 
whose underlying topological order is that of a $\ZZ_2$ gauge theory with a fermionic gauge charge. 
The anomalous property of this SET is encoded in its symmetry fractionalization pattern,
which is easiest to discuss in the framework of X. Chen and M. Hermele~\cite{HermeleChen}.
Namely, dimensionally reducing along one of the boundary directions leads to a quasi-$2+1$D system, 
in which we can make a loop of symmetry defect parallel to the uncompactified directions.
The anomalous property is then that inside the loop we have a $3$-fermion topological order 
whereas outside the loop we have the ordinary toric code.
More generally, the anomaly indicator is that the naive chiral central charges of the two topological orders,
computed mod $8$ from their anyon statistics, differ by $4$.
This is a new kind of anomaly for unitary $\ZZ_2$ symmetry in $3+1$D.

It may be surprising that field theory correctly predicts the existence of a beyond-cohomology phase 
even though it also identifies the GDS dual as a beyond-cohomology phase 
despite the GDS dual's equivalence to the in-cohomology phase in flat space.
One clue may come from our consideration of a ``phase rule'' 
for defining the phase of a given domain wall configuration in our model.
One natural choice of the phase rule that we discuss later comes from the Crane-Yetter TQFT 
and reproduces at least some properties of the GDS dual action,
in particular the dependence of the symmetry of the ground state on the Euler characteristic of the manifold.
However, a different phase rule that we present does not have these properties of the GDS dual action 
but still retains all the anomalous defect properties that we discuss above; 
hence, we believe the mixed gauge-gravity response in the action is not essential to defining the beyond-cohomology phase.

\section{Exactly solved $4+1$D model} \label{sec:model}
\newcommand{\smallR}{{r_{\rm mat}}}
\newcommand{\largeR}{{r_{\rm spin}}}
\newcommand{\Psiprod}{\Psi_{\rm prod}}
\newcommand{\unhealedQCA}{\xi}

Our primary technical contributions are the construction of a commuting projector Hamiltonian for a theory 
that we argue is a 4+1D beyond cohomology SPT and the construction of a circuit $\Udis$ which disentangles the ground state.
The ground state wavefunction $\Psi_0$ is a ``decorated domain wall'' construction.
Roughly speaking, the ground state has fluctuating spin degrees of freedom with additional degrees of freedom 
on the domain walls between spins, 
where the additional degrees of freedom are in the three-fermion Walker-Wang ground state~\cite{WalkerWang}.
See for example Ref.~\onlinecite{chen2014symmetry} for previous work on decorated domain walls, though that work realized
 in-cohomology SPT phases by the decorated domain wall construction; see also Ref.~\onlinecite{cordova2019decorated} which appeared while this paper was in preparation.

A key role will be played in our construction by both quantum circuits and quantum cellular automata (QCA),
so we briefly review the distinction.
A quantum circuit is a unitary $U$ which can be written as a product $U=U_d U_{d-1} \cdots U_1$
where each unitary $U_i$ for $1\leq i \leq d$ is a product of unitaries supported on disjoint sets of bounded diameter;
each of the individual unitaries in the product for $U_i$ is called a {\it gate}.
The index $i$ on $U_i$ labels the {\it round} and the number $d$ is called the {\it depth} and the bound on the diameter of the gates is called the {\it range} of the gates.
Implicitly, when we refer to a quantum circuit throughout this paper, we mean that the depth of the circuit and the range of the gates are both bounded by some $O(1)$ constants, independent of system size.

A QCA $\alpha$ is a $*$-automorphism of the algebra of operators, 
subject to certain locality constraints: 
given any operator $O$ supported on some site, 
the operator $\alpha(O)$ is supported within some distance $R$ (called the {\it range} of the QCA) of $O$.
The term ``$*$-automorphism'' means that $\alpha$ maps operators to operators,
while preserving the product structure and Hermitian conjugation structure, 
i.e., $\alpha(OP+Q)=\alpha(O) \alpha(P) + \alpha(Q)$ and $\alpha(O^\dagger)=\alpha(O)^\dagger$.
For any finite system, any $*$-automorphsim can be written as conjugation by a unitary: 
$\alpha(O)=V O V^\dagger$ for some $V$ depending on $\alpha$, 
and we will sometimes simply say that a unitary ``is a QCA'' to mean that conjugation by that unitary 
obeys the locality requirement of a QCA.  
Hence, every quantum circuit is a QCA, with the range of the QCA bounded 
by some function of the depth of the circuit and the diameter of the gates in the circuit.

However, not every QCA is a circuit.
A QCA that is not a circuit is called {\it nontrivial}.
A standard example is a ``shift" on a one-dimensional line~\cite{Gross_2012}.
Strong evidence has been presented for a three-dimensional nontrivial QCA $\alpha_{WW}$
that it is not even a combination of a quantum circuit and a shift~\cite{haah2018nontrivial}.
This QCA $\alpha_{WW}$ is constructed so that it disentangles the three-fermion Walker-Wang model ground state,
mapping it to a product state with all spins in the $Z$-direction.
More strongly, it maps a specific choice of commuting projector Hamiltonian for the three-fermion Walker-Wang model 
onto a sum of Pauli $Z$ operators on each qubit.
The nontrivial nature of this QCA plays a key role in our description of the effective boundary symmetry below,
and our argument as to why the Else-Nayak construction~\cite{ElseNayak} terminates.

Our Hamiltonian for a 4+1D beyond cohomology SPT is constructed by a unitary $\Udis$ 
which disentangles the decorated domain wall ground state $\Psi_0$, 
so that $\Udis \Psi_0$ is a product state which we write $\Psiprod$.
Since that product state is trivially the ground state of a commuting projector Hamiltonian (indeed, it will be a sum of Pauli operators on each qubit) we can conjugate that trivial commuting projector 
Hamiltonian $\Udis$ to obtain a commuting projector Hamiltonian for our beyond cohomology phase.
The unitary $\Udis$ will be realized by a quantum circuit of bounded depth and range.

This section is organized as follows. 
\cref{sec:geometry} describes the geometry of a cellulation of the four-manifold that we will use, defines the degrees of freedom of the model, and defines the state $\Psiprod$.
\cref{sec:UdisProperties}
gives a general construction of decorated domain walls on a fluctuating configuration of spins and gives defining properties of
the disentangler $\Udis$.
This subsection assumes the existence of a family of unitaries $\Umat$ which, roughly speaking, 
create a particular state on the domain walls, given a configuration of the spins.
The construction of these unitaries $\Umat$ in the particular case that 
we decorate with the three-fermion Walker-Wang model 
is in \cref{sec:3FWWarb,sec:QCAstitch,sec:phase}.
In \cref{sec:3FWWarb} we construct a three-fermion Walker-Wang state model an arbitrary domain wall configuration;
this requires extending the construction of the three-fermion Walker-Wang model state from a three-dimensional square lattice to more general three-manifolds.
In \cref{sec:QCAstitch}, we construct a disentangling QCA for this model.
Finally, \cref{sec:phase}, we fix a phase ambiguity in the unitary defined by this QCA.
An additional section, \cref{sec:CY}, sketches a different way to resolve the phase ambiguity.

\subsection{Geometry}
\label{sec:geometry}

We fix some Voronoi cellulation of the system generated by points in general position.
For each $4$-cell in the Voronoi cellulation, we have one qubit degree of freedom.
One may imagine this degree of freedom as lying somewhere in the center of the $4$-cell.
We refer to these degrees of freedom as the {\it spins}.
On each $3$-cell we have some additional degrees of freedom, which are also qubits; 
in our construction each $3$-cell will necessarily have a rather large number of these degrees of freedom 
($\gg 10^3$ such qubits).
We call these qubits on the $3$-cells the {\it material}.

The state $\Psiprod$ will be a product state with all spins in the $+1$ eigenstate of the Pauli $X$ operator and all material degrees of freedom in the $+1$ eigenstate of the Pauli $Z$ operator.

Note that for any configuration $\vec z$ of spins 
the subset of $3$-cells containing a domain wall forms a closed $3$-manifold $M_{\vec z}$. 
This is the reason for choosing a generic cellulation,
and it is completely analogous to the reason for choosing the double semion model in two dimensions 
to be defined on a cellulation with trivalent vertices (such
as a hexagonal tiling) so that any closed $1$-chain 
is a collection of closed loops without self-intersection.

We will arrange the material qubits within the $3$-cells 
so that (at least locally) they form a cubic lattice.  
We do this in three steps:
First, we triangulate each $3$-cell in some arbitrary way into $O(1)$ $3$-simplices
so that the closed $3$-manifold $M_{\vec z}$ is triangulated (a simplicial complex).
Second, we decompose every $3$-simplex into a union of $4=3+1$ cuboids.
Each cuboid occupies a portion of the $3$-simplex which is closer to a vertex of the $3$-simplex than any other vertex.
All the $4$ cuboids meet at a point in the center of the $3$-simplex.
This procedure is depicted in \cref{fig:cubulation}.

\begin{figure}
\includegraphics[width=\textwidth, trim={5ex, 75ex, 58ex, 0ex}, clip]{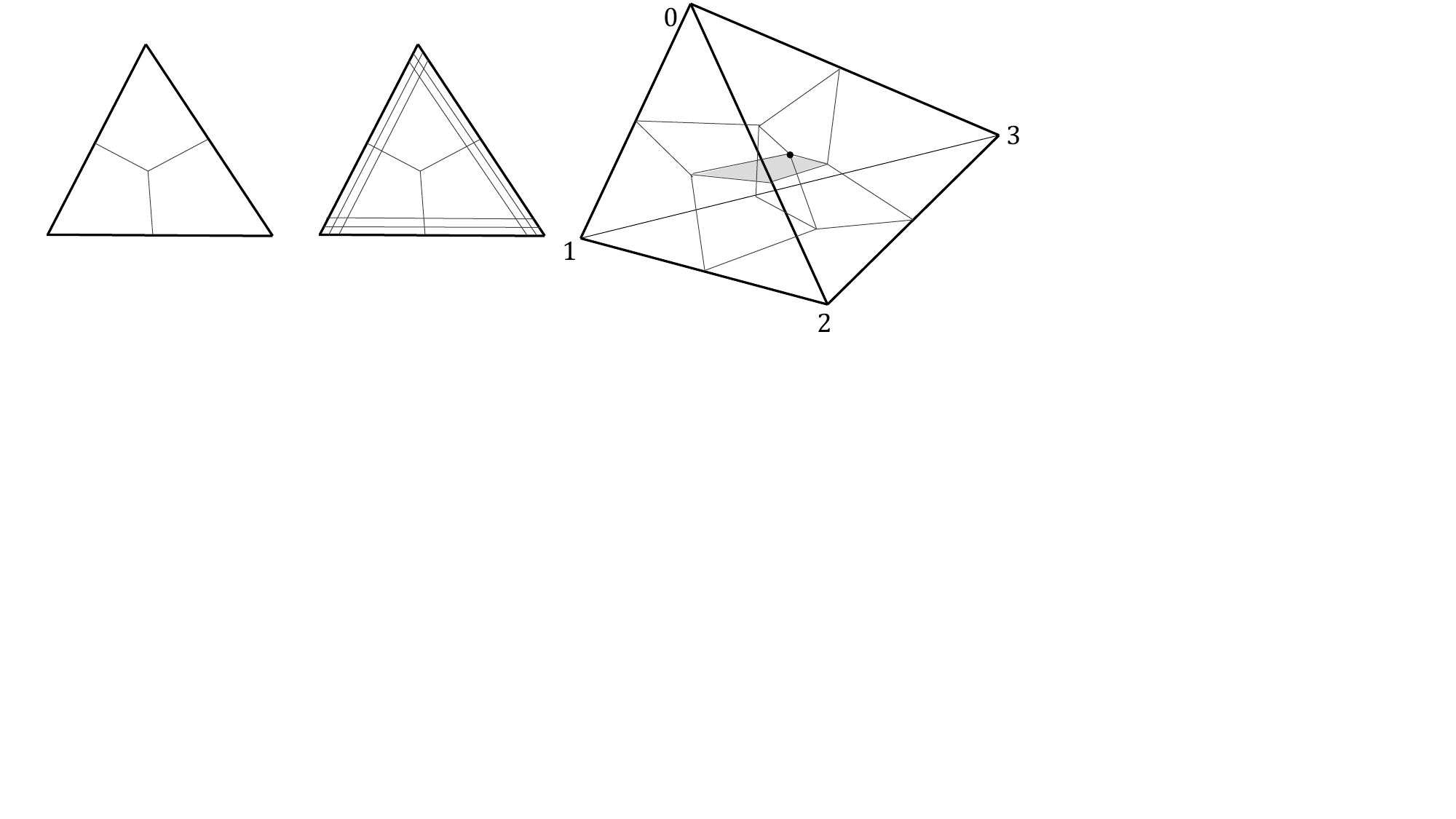}
\caption{Decomposition of simplices.  The leftmost figure shows the decomposition of a $2$-simplex into 3 quadrilaterals and the
third figure shows the subdivision of a $3$-simplex into 4 cuboids.
The black dot in the $3$-simplex lies at the center of the $3$-simplex.
Planes parallel to faces subdivide the simplex giving a refinement of the cubulation, so that near each vertex there are small parallelepipeds
forming a (shear transformed) cubic lattice;
we have depicted this subcubulation only for the $2$-simplex as shown in the second figure.
See \cref{rem:cubulation}.
The shaded quadrilateral in the third figure is $\square_{02}$.}
\label{fig:cubulation}
\end{figure}

This decomposition of $3$-simplices into cuboids gives a cubulation of the $3$-skeleton of the Voronoi cellulation.
Then, in the third step we refine this cubulation by subdividing each cuboid into $\ell \times \ell \times \ell$ subcubes in the obvious way,
where $\ell$ will be chosen sufficiently large later.
Finally, we place the qubits of the material at the edges of this cubulation.

\begin{remark}\label{rem:cubulation}
The topology of the $4\ell^3$ subcubes in a $3$-simplex $\Delta$ can be furnished
with various geometries, but we specify one that will be technically useful later.
Consider the standard $3$-simplex $\Delta$ in $\mathbb R^3$
defined by inequalities $x,y,z \ge 0$ and $x+y+z \le 1$.
Let $\epsilon \in (0,\frac 1 4)$ be a real number.
Bring four sets of planes defined by
\begin{align*}
&x = k \frac{\epsilon}{\ell},\qquad 
y = k \frac{\epsilon}{\ell}, \qquad 
z = k \frac{\epsilon}{\ell}, \qquad
x+y+z = 1 - k \frac{\epsilon}{\ell},
\end{align*}
where $k$ assumes any value among $1,2,\ldots, \ell -1$.
These $4(\ell-1)$ planes define four cubic lattice patches $\Lambda_v$ ($v=0,1,2,3$) inside $\Delta$
near the vertices $v \in \Delta$.
They all consist of small parallelepipeds. 
Six additional quadrilaterals $\square_{vv'}$ ($vv' = 01,02,03,12,13,23$) 
inside the $3$-simplex as depicted in \cref{fig:cubulation},
separate the cubic lattice patches;
the ``divider'' $\square_{vv'}$ sits in between $\Lambda_v$ and $\Lambda_{v'}$.
The divider $\square_{02}$ is shaded in \cref{fig:cubulation}.
Together, they subdivide the $3$-simplex into $4\ell^3$ small cuboids.

The advantage of this geometry is that 
each $\Lambda_v$ is a shear transform of (a finite part of) the standard cubic lattice in $\mathbb R^3$;
only the small cuboids at the intersection of two or three cubic lattice patches,
are not linear transforms of a standard cube in $\mathbb R^3$.
However, the induced subdivisions of the $2$- and $1$-cell in the intersection of cubic lattice patches 
are translation invariant within the $2$- and $1$-cell, respectively. 
$\diamond$
\end{remark}

The $\ZZ_2$ symmetry will act only on the spins of the system;
it will act as the product of Pauli $X$ on all spins, leaving the material unchanged.
By abuse of notation,
we will write $X$ to represent this symmetry operator whenever it is clear.

There are two \emph{microscopic} length scales $\smallR$ and $\largeR$ in our system:
The smaller $\smallR$ is the spacing between nearest neighbor degrees of freedom in the material.
The larger $\largeR = O(\ell \smallR)$ is the spacing between nearest neighbor degrees of freedom in the spins.
They are both microscopic scales that should be regarded as
constants in the limit of large system sizes,
but we sometimes distinguish them for more clear presentation.

\subsection{Defining properties of $\Udis$}\label{sec:UdisProperties}

We construct the unitary $\Udis$ as a controlled unitary, 
controlling its action on the material depending on the configuration of spins in the $Z$ basis.
That is, we write
\begin{align}
\Udis = \sum_{\vec z} \Pi^{spin}_{\vec z} \otimes \Umat_{\vec z},
\end{align}
where the sum is over configurations of spins in the $Z$ basis, written as $\vec z$.
The projector $\Pi^{spin}_{\vec z}$ acts on the spins, projecting onto
spin configuration $\vec z$, while the unitary $\Umat_{\vec z}$ acts on the material, and depends on $\vec z$.

We construct the state $\Psi_0$ so that in the interior of each $3$-cell, 
we have either a three-fermion Walker-Wang state or a product state,
depending on whether there is a domain wall on that $3$-cell,
i.e., there is a domain wall if $Z_i Z_j=-1$, 
where the two spins $i,j$ are in the two $4$-cells attached to that $3$-cell, 
and there is no domain wall if $Z_i Z_j=+1$.
This is achieved by demanding that $\Umat_{\vec z}$ 
should act as the identity on any material degree of freedom 
which is {\it not} in a $3$-cell containing a domain wall.
This includes any material degree of freedom in the interior of a $3$-cell 
if that $3$-cell does not contain a domain wall, 
as well as material in the $2$-skeleton 
so long as all $3$-cells attached to the given $2$-cell do not contain a domain wall.

We demand that $\Umat_{\vec z}$ be a QCA of range $O(\smallR)$ that disentangles
the Walker-Wang state on the three-manifold $M_{\vec z}$
defined by the domain wall of spin configuration $\vec z$, mapping that state to the product state with all material qubits in the $Z=+1$ state.
We demand that $\Umat_{\vec  z}$ depend only on the domain wall configuration, i.e., that
\begin{align}
\Umat_{\vec z} = \Umat_{-\vec z}
\end{align}
where $-\vec z$ is obtained from $\vec z$ by flipping all the spins.
This implies the $\ZZ_2$ symmetry 
\begin{align}
\Udis^\dagger \symX \Udis \symX=I.
\end{align}
Then, since $\Psiprod$ is a $+1$ eigenstate of $\symX$, 
the state $\Psi_0$ also is a $+1$ eigenstate of $\symX$.

We want $\Udis$ to preserve locality when it acts by conjugation;
$\Udis$ is going to be expressed as a circuit whose range is of order $\largeR$.
To this end, we require that
\begin{align}
(\Umat_{\vec z})^\dagger \Umat_{\vec z + i}&\text{ is an operator on the $O(R)$-ball centered at } i \label{eq:UU}
\end{align}
where $\vec z +i$ is a spin configuration obtained from $\vec z$ by flipping spin $i$. 
Furthermore, we require that
\begin{align}
(\Umat_{\vec z})^\dagger \Umat_{\vec z + i}  = (\Umat_{\vec z+ j})^\dagger \Umat_{\vec z + i+j}
\text{ for any $j$ that is $O(R)$-far from $i$.}
 \label{eq:UUUU}
\end{align}
We will show later that all these requirements are satisfied,
but it will be more instructive to see first how these imply that $\Udis$ is a circuit.

\begin{lemma}\label{lem:uqca}
Let $U$ be any controlled unitary
\begin{align}
U=\sum_{\vec z} \Pi^{spin}_{\vec z} \otimes V_{\vec z}
\end{align}
where for every spin configuration $\vec z$ the unitary $V_{\vec z}$ on material acts as a QCA of range $R$ by conjugation.
Then, $U$ acts as a QCA of range $O(R)$ by conjugation if and only if for any spin $i$
\begin{align}
(V_{\vec z})^\dagger V_{\vec z + i} &\text{ is an operator on the $O(R)$-ball centered at }i \label{UU} \\
 \text{ and } \quad (V_{\vec z})^\dagger V_{\vec z + i} &= (V_{\vec z+ j})^\dagger V_{\vec z + i+j} \text{ for any $j$ that is $O(R)$-far from $i$.}\label{UUUU}
\end{align}
\end{lemma}
Note that in this lemma $V(\vec z)$ is \emph{not} necessarily supported on the domain wall determined by~$\vec z$;
it is simply a locality preserving unitary.
\begin{proof}
Define QCA $\beta(\cdot)$ by  $\beta( O ) = U^\dagger O U$.

($\Leftarrow$)
For any unitary~$P_x$ on a single qubit $x$ on the material, 
we have 
$\beta(P_x) = U^\dagger P_x U = \sum_{\vec z} \Pi^{spin}_{\vec z} \otimes (V_{\vec z})^\dagger P_x V_{\vec z}$.
This commutes with any material operator far away from $x$ (since $V_{\vec z}$ acts as a QCA by conjugation) and any spin $Z$ operator.
This also commutes with any spin operator $X_i$ at~$i$ far from~$x$
because
$\beta(P_x^\dagger) X_i \beta(P_x) X_i = \sum_{\vec z} \Pi^{spin}_{\vec z} \otimes 
(V_{\vec z})^\dagger P_x^\dagger V_{\vec z} (V_{\vec z + i})^\dagger P_x V_{\vec z + i} = I$
where $P_x$ and $V_{\vec z} (V_{\vec z + i})^\dagger$ commute by~\eqref{UU}.
So, $\beta(P_x)$ is local.

For any single spin operator~$Z_i$ at~$i$, we see~$U$ commutes with $Z_i$.

For any single spin operator~$X_i$ at~$i$,
we have $\beta(X_i) = \sum_{\vec z} \ket{\vec z + i} \bra{\vec z} \otimes (V_{\vec z + i})^\dagger V_{\vec z}$.
Here, $\beta(X_i)$ commutes with all material operators far from~$i$ by \eqref{UU},
as well as with any spin operator~$Z$.
For any single spin operator~$X_j$ at~$j$ far from~$i$,
we have $\beta(X_i) X_j \beta(X_i) X_j = \sum_{\vec z} \Pi^{spin}_{\vec z + j} \otimes (V_{\vec z + j})^\dagger V_{\vec z + i + j}(V_{\vec z + i})^\dagger V_{\vec z} = I$
by~\eqref{UUUU}.

($\Rightarrow$)
If $U$ acts as a QCA by conjugation,
$\beta(X_i)$ is supported near spin $i$ for a spin Pauli operator~$X_i$ at~$i$.
Hence, since this operator must commute with $X_j$ for $j$ far from $i$, not only is 
$(V_{\vec z+i})^\dagger V_{\vec z}$ an operator on the material supported near spin $i$, 
it is \emph{equal} to
$(V_{\vec z+i+j})^\dagger V_{\vec z+j}$ for any $j$ far from $i$.
In words, 
$(V_{\vec z+i})^\dagger V_{\vec z}$ depends only on spins near $i$.
\end{proof}

Note that
$\Udis$ acting by conjugation as a QCA is a stronger requirement 
than just that $\Umat_{\vec z}$ acts by conjugation as a QCA as it also imposes locality requirements on the spin degrees of freedom.
Heuristically, it means that changing $\vec z$ locally will only change the action of $\Umat_{\vec z}$ locally.

Following~\cite{haah2018nontrivial},
we believe that if one restricts $\Umat_{\vec z}$ to just the material degrees of freedom in the domain wall manifold $M_{\vec z}$,
then it is nontrivial, i.e., it cannot be written as a quantum circuit of depth~$O(\smallR)$
acting just on those material degrees of freedom.
However,~$\Umat_{\vec z}$ can be written as a quantum circuit 
when one acts on {\it all} the material degrees of freedom, not just those in~$M_{\vec z}$.
To see this, recall that we require~$(\Umat_{\vec z})^\dagger \Umat_{\vec z + i}$ to be a local operator.
We can find a sequence of spin configurations,
starting with all entries equal to~$+1$, and ending at a given~$\vec z$, 
differing only by flipping a single spin at a time.
Taking the product of the local operators corresponding to these spin flips,
and doing far separated spin flips in parallel,
gives a circuit for~$\Umat_{\vec z}$.

We use a similar idea to show that $\Udis$ can be realized as a circuit.
\begin{lemma}
\label{lem:circfromcontrolledqca}
Let $U$ be any controlled unitary
\begin{align}
U=\sum_{\vec z} \Pi^{spin}_{\vec z} \otimes V({\vec z}),
\end{align}
such that conjugation by $U$ acts as a QCA with range $R = O(1)$.
Assume that $V(\vec z)=I$ if~$\vec z_i = +1$ for all~$i$.
Then, 
$U$ can be written as a quantum circuit of depth $O(1)$ with gates having range $O(1)$.
\end{lemma}
\begin{proof}
Tile the four-manifold with five different colors of tiles, such that no two tiles of the same color are within distance $2R$ of each other, and such that each tile has diameter $O(R)$.  Call the colors of the tiles $0,1,2,3,4$.
One may derive the tiling from thickening a cell decomposition, with a given color $d$ obtained by thickening $d$-cells.  Indeed, any $O(1)$ number of colors in the tiling would work in the following construction.

For any set $S$ define $\vec z_S$ so that $(\vec z_S)_i=\vec z_i$ for $i\in S$ but $(\vec z_S)_i=1$ for $i\not\in S$.
Let $S_a$ be the set of all spins in tiles colored $b$ for $b\leq a$ so that $S_4$ contains all spins.
Let $S_{-1}=\emptyset$.
Define $U_S=\sum_{\vec z} \Pi^{spin}_{\vec z} \otimes V(\vec z_{S})$.
So, $U_{S_4}=U$.
We will show that $U_{S_{a+1}} U_{S_a}^\dagger$ is a quantum circuit of depth~$1$
with gates of range~$O(1)$ for any $a \in \{-1,0,1,2,3\}$ 
from which the lemma follows by composing these circuits.

To show this claim, we consider any set $S$ containing some tiles and any tile $T$ not in $S$ and show
that $U_{S\cup T} U_S^\dagger$ is a quantum circuit containing a single gate 
supported within distance~$R$ of the tile~$T$.
Then, all the gates for tiles of a single color can be executed in parallel in a quantum circuit,
showing that $U_{S_{a+1}} U_{S_a}^\dagger$ is a quantum circuit.

Let $X_{T,\vec z}$
be the product of $X_i$ over $i\in T$ with $\vec z_i=-1$.
Acting on the subspace with $Z_i=\vec z_i$ for $i\in T$, we have
$U_{S\cup T} U_S^{\dagger}=U_{S\cup T} X_{T,\vec z} U_{S\cup T}^{\dagger} X_{T,\vec z}$, i.e.,
\begin{align}
U_{S\cup T} U_S^{\dagger}=\sum_{\vec z} U_{S\cup T} X_{T,\vec z} U_{S\cup T}^{\dagger} X_{T,\vec z} \Pi^{spin}_{\vec z}.
\end{align}
Then, since conjugation by $U_{S\cup T}$ also acts as a QCA of range $R$ (see next paragraph), this product $U_{S\cup T} U_S^{\dagger}$ is supported within distance $R$ of $T$.  
Hence, since this operator is local, it can be written as a quantum circuit of depth $1$ with a single gate.

To see that conjugation by $U_{S\cup T}$ also acts as a QCA of range $R$,  put $S' = S \cup T$.  $U_{S'}$ is still in the form of \cref{lem:uqca} with $V(\vec z_{S'})$ in place of $V(\vec z)$.
But $V(\vec z_{S'})$ obviously satisfies the two conditions in \cref{lem:uqca}, hence conjugation by $U_{S\cup T}$ acts as a QCA of range $R$ .
\end{proof}

\subsection{Three-fermion Walker-Wang state on arbitrary $3$-manifold} \label{sec:3FWWarb}

Here we consider the three-fermion Walker-Wang (3FWW) model~\cite{WalkerWang,BCFV} on an arbitrary closed $3$-manifold~$M$
that is possibly \emph{non}orientable.%
\footnote{
The orientability is not important for 3FWW,
but our $3$-manifold is a domain wall between spin up and down
and hence is orientable.
}
The original Walker-Wang model~\cite{WalkerWang} 
is defined on a cellulation of an oriented 3-manifold
whose 1-skeleton is a trivalent graph.
The Hamiltonian terms are designed to drive fluctuations of closed string configurations
where string segments are labeled by simple objects of a unitary braided fusion category.
They are decorated by the $F$- and $R$-symbols 
such that each string configuration has correct amplitude (relative to the empty configuration)
when it is interpreted as an anyon fusion/braiding process.
Hence, the trivalency of the 1-skeleton is generally important in the construction since $F$-symbols may be nontrivial.
This trivalency requires that one resolves any high valency vertex of a general cellulation 
into several trivialent vertices, which results in a complicated (albeit systematic) formula for the Hamiltonian terms.

Specializing to the three-fermion theory as input algebraic data to the Walker-Wang prescription,
Ref.~\cite{BCFV} gives a much simpler lattice Hamiltonian in the \emph{flat} cubic lattice in $\RR^3$,
but the verification that this simplified version is what the Walker-Wang prescription gives,
is not presented in~\cite{BCFV}.
There is a technical reason we prefer the version of Ref.~\cite{BCFV} --- 
each term is a tensor product of Pauli operators;
we believe this special property is not necessary in the end,
but our tool to construct a disentangling QCA is bounded by this.

We need to combine the good features of the two different versions: 
we have to deal with arbitrary orientable closed $3$-manifolds $M$
and at the same time we want the terms of the Hamiltonian to be a tensor product of Pauli operators.
To this end, instead of simplifying the complicated formula of the original construction
to show that in some gauge choice the formula gives Pauli operators,
we reinterpret the simpler version from a more topological perspective,
from which the extension to arbitrary orientable $3$-manifolds will be immediate.
Our reinterpretation will be somewhat specific to the three-fermion theory
that has all the $F$-symbols trivial.
All our claims in regards to the construction 
will be proved without referring to Refs.~\cite{BCFV,WalkerWang}.

For clarity, we first consider the situation 
where $M$ is the flat 3-torus, cellulated with small cubes.
We define two lattices, called primary and secondary, where the primary lattice is
the cubic lattice and the secondary lattice is slightly shifted along $(111)$-direction.
The distance the secondary lattice is shifted is smaller than the half of the lattice spacing $\smallR$.
The three-fermion theory has four anyons $\{1,f_1,f_2,f_3\}$ where $1$ means the vacuum.
We take $\{f_1,f_2\}$ be the generators of the fusion group; 
$f_3$ is always regarded as a bound state of $f_1$ and $f_2$.
We assign one qubit per edge in the primary lattice 
and interpret the state of a qubit in $Z$-basis
as the occupation number of $f_1$.
Similarly, we assign one qubit per edge in the secondary lattice
whose basis state is represented by occupancy of $f_2$ on the edge.
The basis of the full Hilbert space is identified with string segments of $f_1$ and $f_2$.
Since we want the ground state to be superposition of closed string configurations,
we put Gauss law terms at each vertex of the primary and secondary lattice.
They are a tensor product of six Pauli-$Z$ matrices in the cubic lattice.\footnote{Our convention exchanges the roles of $X$ and $Z$ Pauli operators, compared to that of Ref.~\cite{BCFV}.}

To define plaquette terms, we need to pick a projection of the 3-dimensional lattice onto a 2-plane.
For definiteness, we choose $(111)$-direction.
If we project a square with all the dangling edges (there are $4\cdot 4 = 16$ of them) attached,
then there are exactly two dangling edges that lie inside the square.
We regard a worldline of $f_1$ along one of these two interior dangling edges
as being twisted upon the loop insertion along the square;
the state acquires $-1$ in the amplitude upon the insertion of the loop 
if there is an odd number of $f_1$ string segments on the interior dangling edges.
A worldline of $f_1$ along any other edge is regarded as being unaffected upon the loop insertion along the square.
The role of the projection is to distinguish these two cases.
The same rule applies for $f_2$ on the secondary lattice.
This is an implementation of the fact that $f_1$ and $f_2$ have topological spin $-1$.
Note that here we are relying on the fact that 
(i) the fusion rule is completely trivial and
(ii) the topological spins for both $f_1$ and $f_2$ are real so we do not have to assign an orientation for twist.
If the edge on the secondary lattice that penetrates the square is occupied,
the loop insertion along the square amounts to braiding $f_2$ around $f_1$, which means that the state acquires the braiding phase, 
namely $-1$, in the amplitude upon the insertion of the loop.  
Again, since the modular $S$-matrix of the $3$-fermion theory is real we do not orient the braiding.

Thus, the prescription for the plaquette term in the primary lattice is 
to take the tensor product of
Pauli $X$'s on the edges of a square, 
one Pauli $Z$ for each of the two edges that lie inside the square upon projection, and
one Pauli $Z$ for the edge on the secondary lattice that penetrates the square.
In the current specific case of the cubic lattice,
a primary plaquette term is a product of $7$ Pauli matrices.
The prescription for the secondary plaquette term is parallel and likewise gives a product of $7$ Pauli matrices.
If we shift the secondary cubic lattice back along $(111)$-direction
to overlay it with the primary lattice,
we recover precisely the Hamiltonian $H_{cubic}$ of Ref.~\citep{BCFV}
as drawn in \cref{fig:BCFV-terms}.

\begin{figure}
\includegraphics[width=0.8\textwidth,trim={1cm 6cm 5cm 3cm},clip]{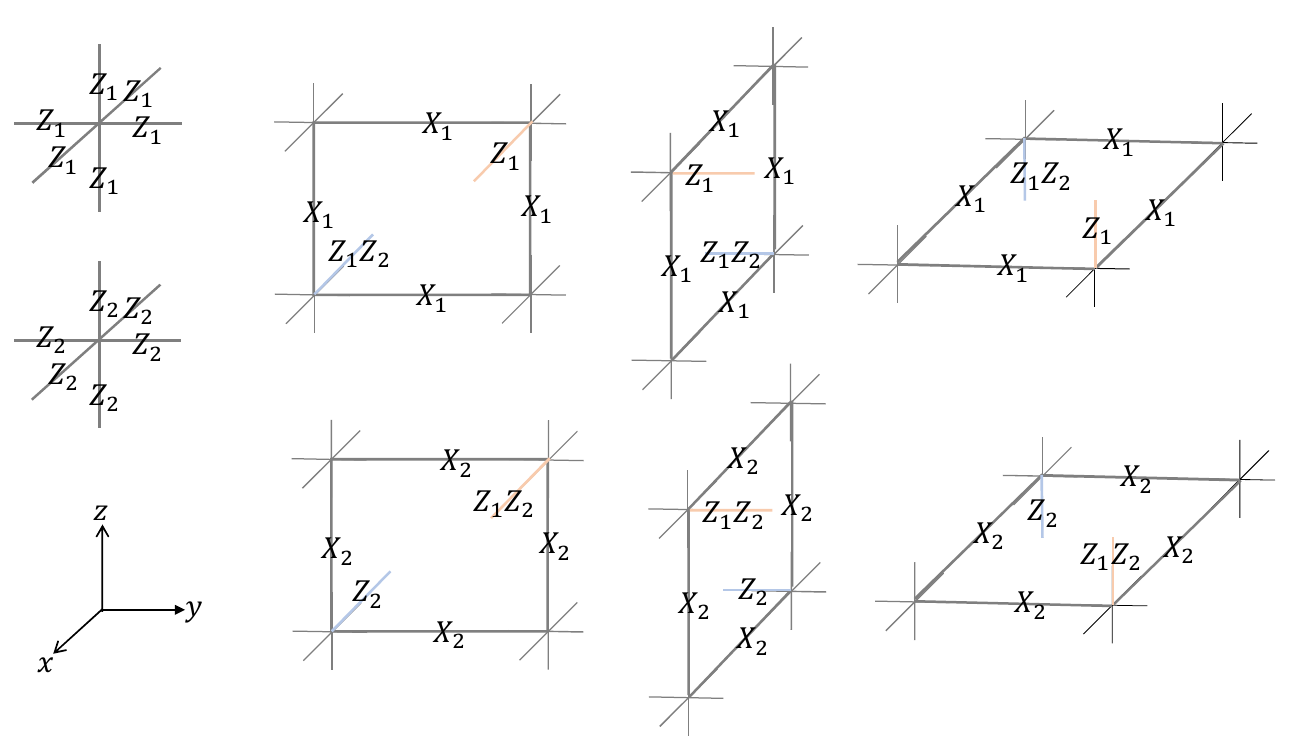}
\caption{
Terms of the 3-fermion Walker-Wang model $H_{cubic}$ on the cubic lattice.
The subscripts $1,2$ distinguish the primary and secondary lattices.
}
\label{fig:BCFV-terms}
\end{figure}

Now we define Hamiltonian terms on an arbitrary closed 3-manifold $M$.
Our general construction will allows us to choose either a triangulation or cubulation of $M$ as the primary lattice.
We pick a piecewise linear map $\phi$ from the $2$-skeleton $M^2$ to $\mathbb R^2$, called the {\it projection},
subject to the condition that $\phi$ be {\em injective} on every $2$-cell and 
that any intersecting pair of distinct edges $e \in f, e' \in f', e \cap e' \neq \emptyset$ 
of adjacent $2$-cells $f,f'$
have {\em transverse} (non-tangential) images under projection.
The projection is an arbitrary choice.
(In the original Walker-Wang prescription~\cite{WalkerWang}, this is done locally by an oriented branching.)
Heuristically speaking,
the projection is a drawing of $M^2$ on a piece of paper.
The injectivity means that the boundary of every $2$-cell must be a closed path without self-intersection,
and the transversality means that we must be able to decide whether a dangling edge to a $2$-cell
that is not contained in the boundary of the $2$-cell 
touches the boundary from outside or from inside.
See \cref{fig:proj-condition}.
The injectivity on every $2$-cell and the transversality implies that 
two $2$-cells that meet along an edge will be projected to one of the configurations in \cref{fig:p-p-commutativity}.
We do not know if such a projection always exists for an arbitrary cellulation.

\begin{figure}
\centering
\includegraphics[width=0.8\textwidth, trim={0ex 80ex 100ex 0ex}, clip]{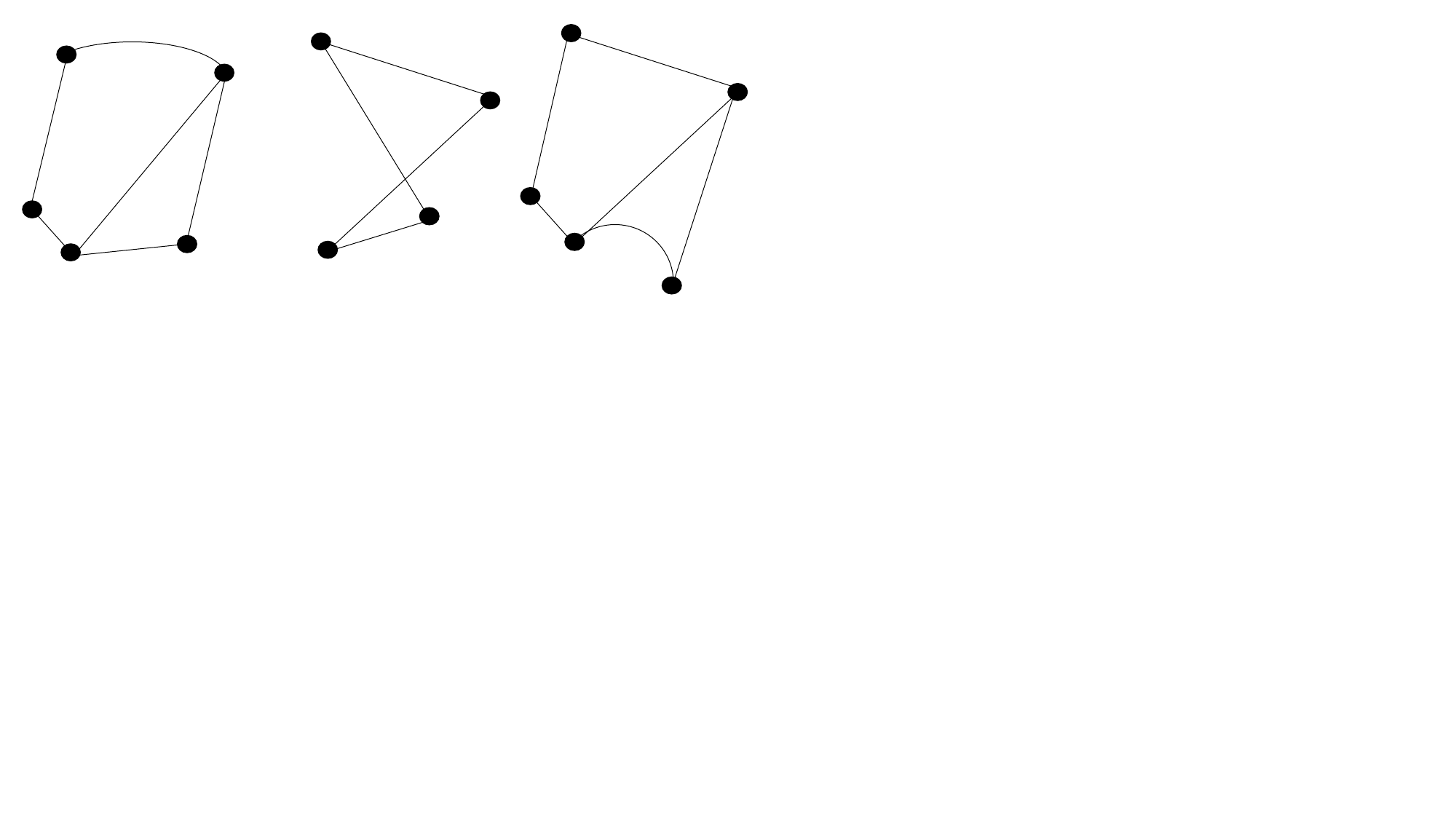}
\caption{Valid and invalid projections. 
A $2$-cell should be projected injectively, which is not the case in the second figure. 
All the adjacent projected edges to a projected $2$-cell 
should meet the boundary transversely, which is not the case in the third figure.
Among these three examples, only the first is valid.}
\label{fig:proj-condition}
\end{figure}

For a triangulation of $M$, the required projection exists as follows.
Choose a generic set of points on $\mathbb R^2$
corresponding to the vertices of the triangulation,
and connect them by straight lines according to $1$-skeleton of the triangulation.
Since any three points are in general positions, a triangle of $M$ is embedded into $\mathbb R^2$.
Using barycentric coordinates, the map for a triangle with vertices $v_0,v_1,v_2$ is 
$\phi( \lambda_0 v_0 + \lambda_1 v_1 + \lambda_2 v_2) = \lambda_0 \phi(v_0) + \lambda_1 \phi(v_1) + \lambda_2 \phi(v_2)$
where $\sum_j \lambda_j = 1$ and $\lambda_j \ge 0$.

For a cubulation that is a refinement of a triangulation according to \cref{sec:geometry},
we choose a particular projection as follows.
First, we consider a valid projection $\phi$ for the parent triangulation, 
as constructed in the previous paragraph.
Second, we extend $\phi$ to the $3$-skeleton of the parent triangulation
by the drawing in \cref{fig:embed-dividers};
using barycentric coordinates, the extension $\phi'$ is defined as
$\phi'(\sum_{j=0}^3 \lambda_j v_j) = \sum_{j=0}^3 \lambda_j \phi(v_j)$.
This is what we visually see in \cref{fig:embed-dividers}.
Third, we cubulate every $3$-simplex following \cref{rem:cubulation}.
The desired projection is now obtained by sending the quadrilaterals in the cubulation by $\phi'$
--- this projection depends on the specific geometry of the cubulation.
To show that this is a valid projection,
we need to check the $2$-cell-injectivity. 
There are two classes of quadrilaterals in the cubulation.
One class consists of parallelograms contained in the planes that are parallel to faces of $3$-simplices.
Since all the additional planes within a $3$-simplex is mapped injectively under $\phi'$,
the injectivity for small parallelograms follows.
The other class of quadrilaterals consists of those on the separating quadrilaterals $\square_{vv'}$ 
that sit in between two cubic lattice patches within a $3$-simplex.
But the dividers $\square_{vv'}$ are already injectively projected as depicted in \cref{fig:embed-dividers},
and so are the small quadrilaterals on them.
Note that the projection squashes the parallelograms along a direction that is fixed within each cubic lattice patch.

\begin{figure}
\includegraphics[width=\textwidth, trim={0ex, 84ex, 65ex, 19ex}, clip]{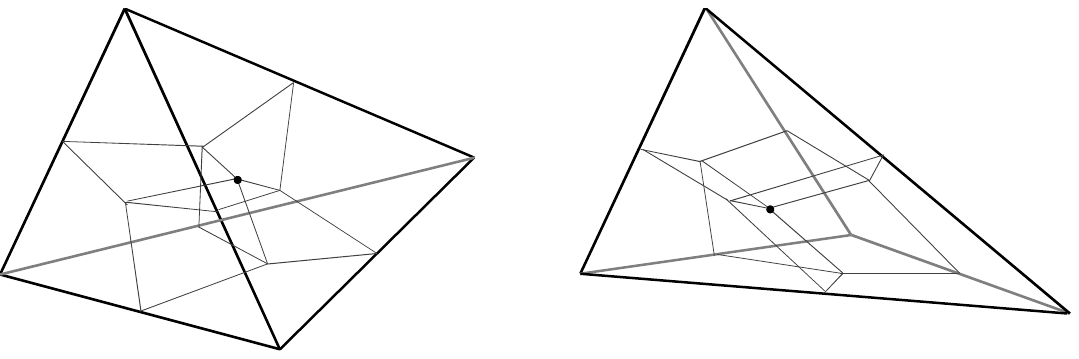}
\caption{Projection of $3$-simplex which is injective on triangles and quadrilaterals.
There are only two projections of a $3$-simplex down to $\RR^2$ 
up to homotopy of the projections such that they are injective on all $2$-faces.
The drawn quadrilaterals $\square_{vv'}$ 
separate cubic lattice patches (not drawn) near vertices
and are injectively projected.
}
\label{fig:embed-dividers}
\end{figure}

Given a primary triangulation, a secondary triangulation is obtained by perturbing it ---
each secondary vertex is at a generic position in a small neighborhood of a primary vertex,
and higher dimensional simplices are placed near those in the primary triangulation.
If a triangulation is refined to become a cubulation as in \cref{sec:geometry},
then we tailor the pertubation as follows.
The vertices of the secondary lattice that belong to the $2$-skeleton of the primary triangulation,
are obtained by a piecewise linear map defined by the perturbation from primary triangulation to secondary triangulation.
This ensures that whatever translation invariance is retained within the interior of $2$-cells of the triangulation.
In the interior of each $3$-simplex,
we push slightly the cubic lattice patch along the squashing direction of the projection.

There is one qubit per primary edge representing worldlines of $f_1$,
and also one qubit per secondary edge for $f_2$.
Every vertex term is a product of Pauli $Z$ around a vertex $v$.
Every primary plaquette term is a product of Pauli~$X$ along the boundary of a primary plaquette~$p$,
decorated with Pauli~$Z$, one on each of the secondary edges that penetrate~$p$,
and another~$Z$ on each of the dangling primary edges (meeting with the plaquette at a vertex)
that lie inside~$p$ upon projection by~$\phi$.
The prescription is analogous for the secondary lattice.
Hence, the Hamiltonian is
\begin{align}
H = 
- \sum_{v} 
\underbrace{
\prod_{e: v \in \partial e} Z_{e} 
}_\text{closed strings}
- \sum_p 
\underbrace{
\left(\prod_{e : e \in \partial p} X_e\right)
}_\text{string insertion}
\underbrace{
\left(\prod_{e': e' \cap p^\circ \neq \emptyset} Z_{e'} \right)
}_\text{mutual braiding}
\underbrace{
\left(\prod_{e : e \cap p = \text{vertex}, \phi (e^\circ) \subset \phi(p^\circ)} Z_e \right)
}_\text{topological spin}
\label{eq:Hamiltonian-3FWW}
\end{align}
where the circle in the superscript means the interior and the prime refers to the secondary lattice.%
\footnote{For an arbitrary cellulation of $M$ equipped with a $2$-cell-injective projection $\phi$,
the criterion whether a dangling projected edge $\phi(e)$ ``lies inside'' a projected plaquette $\phi(\Delta)$
is that there is an open neighborhood $U$ of a vertex $v$ of $\Delta$
such that $\phi(U \cap e \setminus \{v\})$ is contained in the interior of $\phi(\Delta)$.
}
We suppressed the terms for the secondary strings as they are symmetric variants of what are displayed.
Every Hamiltonian term is a tensor product of Pauli $X$ and $Z$ by construction.%
\footnote{
If we omit the $Z$-factors on the dangling edges that account for the topological spin, 
we obtain a Walker-Wang model for the toric code input category $\{1,e,m,\epsilon\}$.
}

A feature of the plaquette terms
is that $X$-factors lie precisely along the boundary of a $2$-cell 
(in the primary or secondary lattice).
Hence, in the multiplicative group of all Hamiltonian terms 
(without the overall minus sign in \cref{eq:Hamiltonian-3FWW}),
every member has $X$-factors along a closed $1$-chain
that is the boundary of some $2$-chain over $\ZZ_2$.

\paragraph*{On the cubulation we have $H_{cubic}$ of \cref{fig:BCFV-terms} in the interior of $3$-cells.}
We have furnished our cubulation (that is a refinement of a triangulation)
with a projection that squashes parallelograms of the cubic lattice patch along a fixed direction within that patch.
In addition, each cubic lattice patch of the secondary lattice 
is by definition a near-identity shift of that of the primary lattice along the squashing direction.
Hence, our prescription for the Hamiltonian gives the same terms as in $H_{cubic}$
for the interior of each $3$-cell.

\begin{figure}[b]
\includegraphics[width=0.9\textwidth, trim={2ex 95ex 57ex 19ex}, clip]{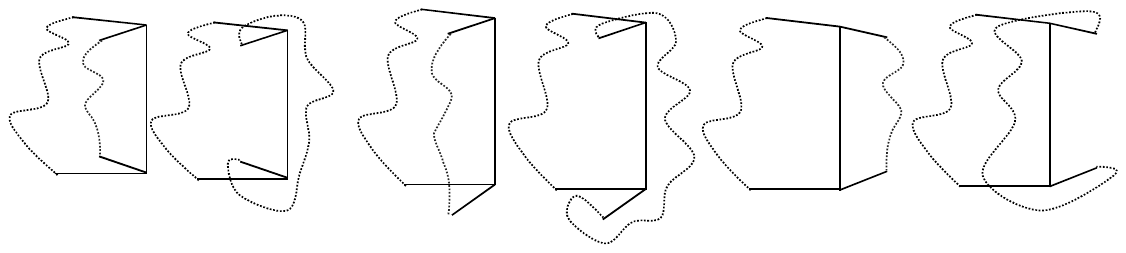}
\caption{Possible configurations of two projected plaquettes sharing an edge.
If a triangulation is projected with straight edges, 
then only the first, third, and fifth configurations are possible.
On the common edge, the two plaquette terms have $X$-factors that commute.
In all cases, either two $Z$-factors from one plaquette term meets two $X$-factors from the other plaquette term,
or a factor of $Z\otimes X$ from one plaquette meets a factor of $X \otimes Z$ from the other plaquette,
or no $Z$-factors of a plaquette meet an $X$-factor of the other.
}
\label{fig:p-p-commutativity}
\end{figure}

Let us show other important properties of the Hamiltonian in \cref{eq:Hamiltonian-3FWW}.

\paragraph{All the terms in the Hamiltonian commute with one another.}
The vertex terms are obviously commuting.
Any vertex term commutes with any plaquette term
since exactly two $X$-factors of a plaquette meet two $Z$-factors of a vertex term.
A plaquette of the primary lattice and another of the secondary lattice
may be linked, be unlinked but their interiors intersect, or not meet.
If they are linked, there are exactly two $Z$ factors that lie in the interiors of the plaquettes,
which meet two $X$ factors on the perimeters of the other plaquettes.
If the plaquettes are unlinked but their interiors intersect, then two edges of one plaquette penetrate the interior of the other plaquette.
The two edges have an $X\otimes X$ factor, and the plaquette has a $Z \otimes Z$ factor, so the two plaquettes commute.
Two plaquette terms in the primary lattice may intersect along an edge or at a vertex.
If they meet along an edge, upon projection they are in one of the configurations in \cref{fig:p-p-commutativity},
in each of which they commute.
If they meet at a vertex, we encounter configurations similar to
\cref{fig:p-p-commutativity} but with the common edge collapsed to the common vertex,
and they still commute.
Two plaquette terms in the secondary lattice commute by the same reasoning.

\paragraph{The Hamiltonian of \cref{eq:Hamiltonian-3FWW} is unfrustrated 
and moreover the vacuum state is a nonzero component in the ground state.}\label{para:vacuum}
That is, each term (without the overall $-$ sign) assumes eigenvalue $+1$ on the ground state
that has expansion
\begin{align}
\Psi(M) = \sum_{g \in S(M)} g \ket 0. \label{eq:3FWW-GS-on-M}
\end{align}
Here, $S(M)$ is the multiplicative abelian group, called the Pauli stabilizer group on $M$,
generated by all the terms of $H$ (without the overall $-1$) and $\ket 0$ is the product state with all qubits in the vacuum state $Z=+1$.
If we just had the fact that each operator in $S(M)$ assumes $+1$ on $\Psi(M)$, without the fact that the vacuum state had nonzero amplitude in the ground state,
then we would not be able to write the expansion 
since $\ket 0$ might be annihilated by the projector $\Pi_{\Psi(M)} = |S(M)|^{-1} \sum_{g \in S(M)} g$.

To prove both claims, we examine the ``diagonal'' subgroup of $S(M)$
consisting of all elements that are products of $Z$'s.
Thanks to the $\ZZ_2$ homological interpretation of plaquette terms,
we see that the diagonal subgroup is generated by products of the plaquette terms over $2$-cycles with $\ZZ_2$ coefficients
and the vertex terms.
Let us define a \emph{sign} for each of the diagonal stabilizers
as the eigenvalue of the diagonal stabilizer on $\ket 0$.
We claim that this sign is always positive,
which implies $\Pi_{\Psi(M)}$ does not annihilate $\ket 0$, proving both claims.
Note that the sign is a group homomorphism from the diagonal subgroup to $\{ \pm 1\}$.
Restricted to plaquette terms,
we may say that the sign is a group homomorphism from the set of all $2$-cycles to $\{\pm 1\}$.
This relies on the commutativity of the plaquette terms, which in turn depends only on the $2$-cell-injectivity of $\phi$.

The sign is clearly positive for vertex terms.
Let us compute the sign for a product $\prod Q$ of the plaquette terms over a $2$-cycle $C_2$ on the primary lattice.
The plaquette terms may have $Z$-factors on the penetrating secondary edges.
These do not contribute to the sign since they do not meet any $X$-factors.
The product $\prod Q$ may also have $Z$-factors on dangling primary edges,
but only those within the support of $C_2$ are important
since, otherwise, they do not meet any $X$-factors.
Hence, we only have to keep the primary edges that are within $N^2 = \mathrm{Supp}~C_2$.
That is, the sign of $\prod Q$ is equal to the sign of the product of plaquette terms 
as if the whole lattice were $N^2$ equipped with a restricted $\phi$ on $N^2$.

Consider the cone over $N^2$, which is a $3$-dimensional simplicial complex $M'$ 
that is not necessarily a manifold.
We extend $\phi|_{N^2}$ to the $2$-skeleton of $M'$.
This is easy: on the plane where $N^2$ is projected, we bring an additional vertex at a generic position,
and connect it to the vertices of $\phi(N^2)$.
By genericity, this defines an extension of $\phi|_{N^2}$ where it is injective on every $2$-cell.
So, we have commuting plaquette operators on $M'$.
Now, the boundary of the cone is $C_2$ since $C_2$ is closed.
Since the sign is a group homomorphism, 
the sign of $\prod Q$ is the product of signs of the boundaries of all $3$-simplices of $M'$.
But the sign of the boundary chain of a $3$-simplex is always positive
as seen by direct calculation.
There are only two possible projections of a $3$-simplex,
distinguished by whether a vertex is projected inside a triangle --- see \cref{fig:proj3cell}.
The same argument goes for cubulation with three different cones over a quadrilateral.
This proves that the sign is always positive.

\begin{figure}[b]
\includegraphics[width=0.9\textwidth, trim={5ex 92ex 25ex 24ex}, clip]{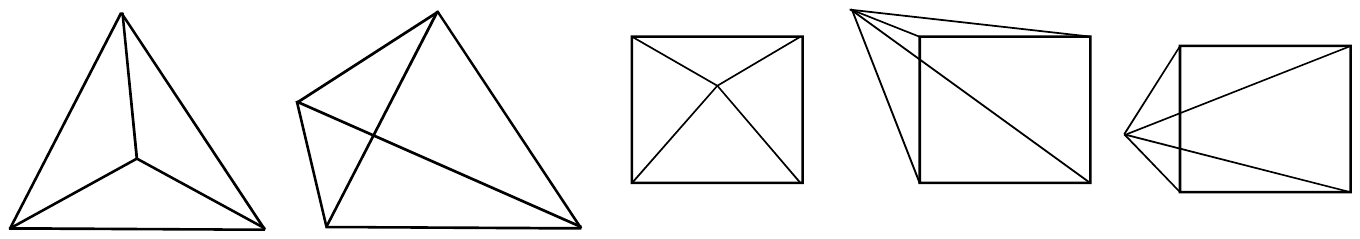}
\caption{Possible projections of a $3$-simplex
and possible projections of a cone over a quadrilateral.
The product of the four or five plaquette terms according to the shown projections
is that of $Z$-factors over all edges in the interior of the enclosing polygon
with the overall~$+$ sign.
An effective and general way to calculate the products is to regard the $2$-sphere (=the union of all plaquettes)
as the union of two hemispheres glued along the very outer rims.
}
\label{fig:proj3cell}
\end{figure}

We summarize the result here as a lemma.
\begin{lemma}\label{lem:1-cycle-to-amplitude}
Let $M$ be a compact simplicial complex or cubulated complex (not necessarily a manifold)
equipped with a continuous map $\phi$ from the $2$-skeleton of $M$ to $\mathbb R^2$ 
such that $\phi$ on every $2$-cell is injective 
and any pair of edges of adjacent $2$-cells have transverse images.
For any $1$-cycle $C_1$ over $\ZZ_2$,
let $\ket{C_1}$ be the basis state specified by the primary string configuration~$C_1$.
Then, the function $C_1 \mapsto \bra{C_1} \prod Q \ket 0 = \pm 1$ from any nullhomologous cycle $C_1$ to amplitude is well defined,
where $\prod Q$ is the product of the primary plaquette terms of the stabilizer group $S(M)$ over any $2$-chain whose $\ZZ_2$-boundary is $C_1$.
\end{lemma}
\begin{proof}
The arbitrariness in the definition of the function is in the choice of the $2$-chain.
But different $2$-chains differ by a $2$-cycle over which the product of plaquette terms has positive sign.
\end{proof}

\paragraph{The Hamiltonian obeys the local topological order condition.}
For the present commuting Pauli Hamiltonian,
the condition is met if any Pauli operator $P$ on a ball of radius smaller than the injectivity radius of $M$
which commutes with every term of the Hamiltonian,
is a product of terms of the Hamiltonian each of which is supported near the ball up to a global phase 
factor~\cite{BravyiHastingsMichalakis2010stability}.
In our case, this condition is satisfied as seen by invoking the homological interpretation of the terms.
The commutativity with the vertex terms means that $X$-factors of $P$ must form a closed $1$-chain over~$\ZZ_2$.
Being contained in a contractible ball, the chain is a homological sum of plaquettes in the ball.
Hence, we are reduced to the case where $P$ is a product of $Z$-operators.
Then, the commutativity with the plaquette terms means that $Z$-factors must form a dual $2$-cycle.
Again, being contained in a contractible ball, 
the dual cycle is a homological sum of small dual $2$-spheres,
which are precisely the vertex terms.
Therefore, our $H$ in \cref{eq:Hamiltonian-3FWW} satisfies the local topological order condition.

\paragraph{There is no deconfined topological charge.}
Here, a deconfined topological charge means
$O(1)$ number of flipped terms that are not creatable by operators on the $O(1)$-neighborhood of the flipped terms.
For our Hamiltonian $H$ in \cref{eq:Hamiltonian-3FWW} 
that consists of commuting Pauli operators,
it suffices to consider Pauli operators acting on the ground state.
Let $P$ be a Pauli operator that creates a hypothetical topological charge at position $x$;
all other excitations by $P$ are far away from $x$, but the number of them is $O(1)$.

Remark that the product $\prod Q$ of primary plaquette operators over a $2$-cycle $C_2$
consists purely of $Z$-factors and
defines two dual $2$-cycles of $Z$-factors due to the commutativity with $X$-part of plaquette terms.
One is formed by penetrating secondary edges, denoted by $b$,
and the other is formed by dangling primary edges, denoted by $a$.
Let us give a \emph{co}homology interpretation for these cycles;
the chains $a^1$ and $b^1$ define $1$-cocycles by commutation relations with Pauli $X$-operators.
First, they are coclosed because the Hamiltonian terms are commuting.
Indeed, coclosedness means that the cochains give zero for boundaries,
i.e., the cochains $a^1$ and $b^1$ as operators commute with $X$-operators on the boundary of any $2$-chain,
which are precisely the $X$-part of the product of the plaquettes over the $2$-chain.
Second, the cocycle $a^1$ is always trivial in the cohomology $H^1(M;\ZZ_2)$.
Indeed, any $1$-cycle of $X$-operators on the primary complex
intersects dangling edges of $C_2$ with $Z$-factor on them an even number of times,
because it does so at every intersection point ---
when a line meets a disk but does not end on the disk, 
there are exactly two edges of the line dangling at the intersection point 
which lie ``inside'' the disk upon our projection.
So, the cocycle $a^1$ as a $\ZZ_2$-functional on $1$-cycles is trivial.
Third,
the cocycle class $[b^1] \in H^1(M;\ZZ_2)$ is trivial if and only if $[C_2] \in H_2(M;\ZZ_2)$ is trivial.
Indeed, if $[b^1] = 0$, 
then $b^1$ as a $Z$-operator is a product of secondary vertex terms since they are generators for null-cohomologous cocycles.
In turn, this means that $C_2$ is the $2$-cycle that encloses these vertices with which the vertex terms are associated.
If $[C_2] = 0$, then $\prod Q$'s secondary lattice part 
decomposes into $Z$-factors on fanning out edges from vertices,
so $b_1$ as a $Z$-operator is a product of those vertex terms.
This is a version of the Poincar\'e duality $H^1(M;\ZZ_2) \cong H_2(M;\ZZ_2)$.
Analogous statements also hold for the product of secondary plaquettes.

Since $P$ is a Pauli operator, it can be written as $P = \mathsf X \mathsf Z$ 
up to an unimportant phase factor 
where $\mathsf X$ is a product of $X$'s only and $\mathsf Z$ is a product of $Z$'s only.
In the region without excitations,
$P$ should commute with any vertex terms and in particular $\mathsf X$ must form strings
with potential end points near the excitations.
Consider a $2$-sphere that encloses the excitation at~$x$ but no other excitations.
If we take the product of all secondary plaquettes on this $2$-sphere,
then by the above remark we have a $1$-cocycle on the primary complex
that can detect the intersection parity of primary $X$-strings.
The primary $X$-strings of $\mathsf X$ must pass the $2$-sphere an even number of times.
The same argument goes for secondary $X$-strings.

Multiplying $P$ by terms of the Hamiltonian does not change the configuration of the excitations.
So, we may assume that $\mathsf X$ does not have any string that connects $x$ and other excitations.
This means that $\mathsf X$ consists of some closed $1$-chains $p$ 
of the primary complex and $s$ of the secondary complex
that stay away from all the excitations
and some $X$-factors near the excitations.
We must have $[p] = [s] = 0 \in H_1(M;\ZZ_2)$ for the following reason.
Since the $O(1)$-neighborhood $E$ of all the excitations have trivial homology $H_1(E;\ZZ_2) = H_2(E;\ZZ_2) = 0$,
we know that the relative homology group $H_2(M,E;\ZZ_2)$ is equal to $H_2(M;\ZZ_2)$.
That is, we can always ``detect'' $p$ and $s$ by some product of plaquette operators.
Hence, we may further assume that $\mathsf X$ is supported near the excitations.

Then, the commutativity of $P$ with plaquettes constrains $\mathsf Z$
such that viewed as a dual $2$-chain, $\mathsf Z$ has no boundary except near the excitations.
Since $H_2(M,E;\ZZ_2) = H_2(M;\ZZ_2)$,
various products over $2$-cycles of plaquette terms represent all classes of $H_2(M,E;\ZZ_2)$
and vertex terms generate all nullhomologous dual $2$-cycles.
Therefore, $P$ is equivalent (up to terms of the Hamiltoniafn)
to local operators near the excitations,
meaning that all the excitations are locally created.

\subsection{Disentangling the 3FWW state on arbitrary $3$-manifold}
\label{sec:QCAstitch}
Here we construct a QCA $\alpha_{WW}(M)$ that disentangles the $3$-fermion Walker-Wang ground state above.
Our $\alpha_{WW}(M)$ is rather involved technically
as the disentangling QCA for the 3FWW on a flat $3$-space~\cite{haah2018nontrivial} was already complicated.
Our construction is an existence proof but it can, at least in principle, be made constructive.
We show how to stitch QCA on local flat $3$-cells together
to make a QCA $\unhealedQCA$ on a given arbitrary $3$-manifold,
and then modify $\unhealedQCA$ by a quantum circuit to completely disentangle the ground state.

\subsubsection{Sewing QCA on cubulations}

Digressing from the construction of $H$ in \cref{eq:Hamiltonian-3FWW} above,
here we define a QCA $\unhealedQCA(M)$ on an arbitrary cubulated closed $3$-manifold $M$
whose action on every local cubic lattice patch agrees with 
the QCA $\alpha_{WW}(T^3)$~\cite{haah2018nontrivial}
that disentangles $H_{cubic}$ of~\cref{fig:BCFV-terms}.
We use the geometry of \cref{rem:cubulation}.
Each of the cuboids of \cref{fig:cubulation} contains
a cubic lattice patch
of linear size that is comparable to, but smaller than $\largeR$.
Each cubic lattice patch 
comes with an arbitrary but fixed frame (a linear basis of the tangent space).
The frame comes from the geometry of \cref{rem:cubulation}: we can undo the shear transform so that the cubic lattice is generated by integer combinations of three orthogonal vectors which give a basis; we then identify those vectors with standard basis vectors $(1,0,0),(0,1,0), (0,0,1)$ in such a way that the projection is in the $(111)$ direction.
With respect to this frame, inside each cubic patch we have the Hamiltonian considered in Ref.~\onlinecite{haah2018nontrivial} which may be disentangled by the QCA constructed in that reference.

A QCA being an $*$-algebra homomorphism is defined once we specify the action on local Pauli operators.
Hence we let $\unhealedQCA(M)$ act by $\alpha_{WW}(T^3)$ of range $O(\smallR)$ on the interior of any cubic lattice patch.
Here the action of $\alpha_{WW}(T^3)$ is defined according to the chosen frame.
The nontrivial point of constructing $\unhealedQCA(M)$ is to show that 
this partial definition of $\unhealedQCA(M)$ extends to all operators on $M$.
Below we extend the definition, in order, on $2$-cells, $1$-cells, and then $0$-cells.
These cells have linear dimension of order $\largeR$ 
and do not refer to the smallest cubes of linear dimension $\smallR$;
this usage is different from that of the previous sections 
where ``cells'' were individual simplices or cubes of a cellulation
and had linear dimension $\smallR$.
Rather, in this subsubsection each $3$-cell refers to a cuboid containing a single cubic lattice patch,
the $2$-cells are between two cubic lattice patches
and the $0$- and $1$-cells are surrounded by many cubic lattice patches.
Since the operator algebra on $M$ is some simple algebra of finite dimension,
we do not have to check its invertibility as it follows automatically:
any nonzero $*$-homomorphism has to be injective since its kernel is an ideal of the simple algebra.

Before we show the extension
we note a few properties related to $\alpha_0 = \alpha_{WW}(T^3)$~\cite{haah2018nontrivial}.
\begin{enumerate}
\item[(i)] $\alpha_0$ is translation invariant (TI) 
and its dimensionally reduced 1-dimensional QCA along any direction is a quantum circuit without shift.
\item[(ii)] $(\alpha_0)^{\otimes 2}$ is a quantum circuit~\cite[Thm.IV.9]{haah2018nontrivial}, 
and hence $\alpha_0^2$ is a quantum circuit~\cite[\S 2]{FreedmanHaahHastings2019Group}.
\item[(iii)] $\alpha_0$ maps a Pauli operator to a tensor product of Pauli operators (Clifford).
\item[(iv)] Every TI Clifford QCA without shift (in the sense of (i)) 
that disentangles the ground state of $H_{cubic}$ in \cref{fig:BCFV-terms},
is equivalent to $\alpha_0$ up to a TI Clifford circuit~\cite[Cor.IV.6 \& Lem.IV.10]{haah2018nontrivial}.
\item[(v)] A one-coordinate inversion~\cite{FreedmanHaahHastings2019Group} of any TI Clifford QCA
is equivalent to its inverse up to a TI Clifford circuit~\cite{haah2019clifford}.
\end{enumerate}
Indeed, the explicit QCA displayed in the supplementary material of Ref.~\cite{haah2018nontrivial} does not have any shift upon dimensional reductions.
In fact, any TI Clifford QCA can be made so by composing with some shift~\cite{haah2019clifford}.
The square being a quantum circuit is a general property of TI Clifford QCA over qubits, given the no shift property (i).
Because of (iii) the quantum circuits in (i) and (ii) are Clifford.
The one-coordinate inversion of a QCA means the one obtained by conjugating the QCA by a spatial inversion about a plane.

Denote by $M^k$ the $k$-skeleton of the cubulated $M = M^3$.
So far, we have defined $\unhealedQCA(M^3 \setminus M^2)$.

\paragraph*{Extension to $M^3 \setminus M^1$ --- filling up $2$-cells:}

Since $\alpha_0$ is defined on a cubic lattice,
we may consider $48$ versions by conjugating it with the lattice symmetry group.
The lattice symmetry group is generated by permutations of three axis and inversions about three coordinate planes.
A nice feature of $H_{cubic}$ is that it is invariant under the permutations of three axis.
Then, by the property (iv) all $6$ versions of $\alpha_0$ is equivalent to itself up to some TI Clifford circuit.
Combining the properties (ii) and (v), 
we see that all the inversions of $\alpha_0$ are equivalent to itself up to some TI Clifford circuit.
Therefore, all the $48$ versions of $\alpha_0$ by cubic lattice symmetry group
are related by some TI Clifford circuits.

Given two versions $\alpha_0'$ and $\alpha_0''$ of $\alpha_0$,
if we keep the gates of the relating Clifford circuit on the half space, say $x > 0$,
but drop them on the complementary half space,
then on the half space where $x < - R$ where $R$ is the range of $\alpha_0$ we have the identity action,
and on the other half space where $x > R$ we have the action of $\alpha_0'' \circ (\alpha_0')^{-1}$.
Composing it with $\alpha_0'$ we obtain a QCA $\gamma$ of range $O(\smallR)$ 
that interpolates between $\alpha_0'$ and $\alpha_0''$.
The interpolating $\gamma$ is still translation invariant along $y$- and $z$-axes.

Now, given a pair of cubic lattice patches in $M$,
we have two versions of $\alpha_0$ according to the frames of the respective cubic lattice patches.
Though it is meaningless to speak of which versions they are individually, 
it is meaningful to speak of a version relative to the other,
and we can apply the construction of the interpolating QCA.
On the interior of the bordering $2$-cell,
we define $\unhealedQCA(M^3 \setminus M^1)$ by the interpolating QCA;
thanks to the translation invariance within the interpolating plane,
this extension can be left undefined on the boundary of the $2$-cells.

Note that by the property (ii),
the interpolating QCA $\gamma$ can be chosen such that $\gamma^{\otimes 2}$ is a Clifford circuit;
we do make such a choice.
The choice is unique up to Clifford circuits in the intermediate region $|x| \le R$.
Indeed, two interpolations may differ by a QCA on the interpolating region that is 2-dimensional,
and such a QCA is always a Clifford circuit followed by a shift~\cite{FreedmanHastings,haah2019clifford}.
In short, the action of two copies of our extension $\unhealedQCA(M^3 \setminus M^1)$ 
can be implemented by a Clifford circuit.

\paragraph*{Extension to $M^3 \setminus M^0$ --- filling up $1$-cells:}

For $1$- and $0$-cells,
we consider a QCA as a collection of mutually commuting simple algebras
that are images of single-qubit algebras.
In the present case of Clifford QCA,
we consider a collection of pairs of anticommuting local Pauli operators
where any two operators in two different pairs commute.
Given such a collection, we can recover the QCA up to a shift QCA and a Clifford circuit of depth~$1$
by assigning each anticommuting pair to a qubit in the support of the pair.%
\footnote{
See the application of the Hall marriage theorem in \cite{haah2018nontrivial} 
or a slightly refined version in \cite[App.A]{haah2019clifford}.
}
To each qubit in the $1$-skeleton $M^1$ of $M = M^3$ where $\unhealedQCA(M)$ needs to be defined,
we have to assign a pair of anticommuting Pauli operators
that commute with any other operator in the image of $\unhealedQCA(M^3\setminus M^1)$.
Let us collect all the candidates.

The commutant of the image $\mathcal B$ of $\unhealedQCA(M^3\setminus M^1)$ is supported on $M^1$.
(Strictly speaking, it is on the $O(\smallR)$-neighborhood of $M^1$,
but we suppress this small discrepancy for clarity of presentation.)
As we did for $2$-cells, we temporarily focus on one $1$-cell surrounded by many $2$- and $3$-cells.
We regard these cells are large along the extended direction of the $1$-cell, say $x$-direction.
The translation invariance of the interpolating QCA within the $2$-cells
means that the commutant~$\mathcal A$ on this $1$-cell is also translation invariant along the $x$-direction.

It is shown~\cite[Thm.IV.11]{haah2018nontrivial} that there exists 
a translation invariant, locally generated, maximal abelian subgroup $\mathcal A_0$
in any group of finitely supported Pauli operators that is translation invariant in 1D.%
\footnote{
There is a weak translation symmetry breaking in the choice of a maximal abelian subgroup,
but the degree of the symmetry breaking is determined
by specifics of $\alpha_{WW}(T^3)$ and the arrangement of $2$- and $3$-cells around the given $1$-cell,
which are fixed once and for all.
}
Furthermore, any such abelian subgroup of Pauli operators becomes, after a Clifford circuit~$C$,
a group generated by single qubit operators $Z$ and two-qubit operators $ZZ$~\cite[\S 6,Thm.3]{Haah2013}.
If~$Z_x Z_{x+c} \in C \mathcal A_0 C^\dagger$ with~$c \neq 0$, then by translation invariance
we have $Z_x Z_{x+nc} \in C \mathcal A_0 C^\dagger$ for arbitrary $n \in \ZZ$, 
but since $\mathcal A$ is the commutant of some locally generated algebra,
we must have $Z_x \in C \mathcal A_0 C^\dagger$.

Let $\{ C^\dagger Z_x C \}_x$ be the generating set for the maximal abelian subgroup of $\mathcal A$.
Certainly, the Pauli operator $\tilde X_x = C^\dagger X_x C$ on the $1$-cell anticommutes with $\bar Z_x = C^\dagger Z_x C$.
The operator $\tilde X_x$ may or may not belong to $\mathcal A$.
However, if there is a Pauli operator, say $\unhealedQCA(M^3\setminus M^1)(Z_i)$, in $\mathcal B$ that does not commute with $\tilde X_x$,
then $\tilde X_x \cdot \unhealedQCA(M^3\setminus M^1)(X_i)$ commutes with $\unhealedQCA(M^3\setminus M^1)(Z_i)$.
The modification $\tilde X_x \to \tilde X_x \cdot \unhealedQCA(M^3\setminus M^1)(X_i)$ does not change the commutation relation with $\bar Z_{x'}$ for any $x'$.
Since $\tilde X_x$ is a local operator, 
some $O(1)$ number of similar modifications give a Pauli operator $\bar X_x \in \mathcal A$ 
that anticommutes with and only with $\bar Z_x$.%
\footnote{
An argument of a similar spirit appears in \cite[Lem.3.9]{FreedmanHastings}.
}
The translation invariance applies to this procedure,
and we obtain a finite set $\{ (\bar X_{x_j}, \bar Z_{x_j}) : j = 1,\ldots, m \}$ of anticommuting pairs
that generates $\mathcal A$ by translations.

Now we want to assign the found pairs $\{(\bar X_x, \bar Z_x)\}_x$ 
to Pauli operators on $M^1$ on which $\unhealedQCA(M)$ is not defined yet.
For this assignment to be locally feasible
we have to ensure that there are the same number of the pairs in $\mathcal A$
and unassigned qubits in $M^1$.
To compare these numbers, we take two copies of the system.
(We may impose periodic boundary condition along $x$-direction to remain in a finite dimensional setting.)
Recall that we have chosen our extension $\unhealedQCA(M^3 \setminus M^1)$
such that its double is a Clifford circuit.
By applying the inverse of this Clifford circuit to the doubled system
the algebra $\mathcal B \otimes \mathcal B$ becomes that of the tensor product of individual qubit algebras (trivial).
The commutant of the trivialized $\mathcal B^{\otimes 2}$ is again trivial,
the two numbers of interest are the same.

We assign each pair $(\bar X_{x_j},\bar Z_{x_j})$ in the finite generating set of $\mathcal A$
to an unassigned qubit in its support, and extend the assignment by translation.
This defines an extension $\unhealedQCA(M^3 \setminus M^0)$.
As in the case of filling up the $2$-cells,
any two possible extensions to the interior of a $1$-cell differ by a QCA on that interior of the $1$-cell.
We fix the ambiguity by requiring that the action of $(\unhealedQCA(M^3 \setminus M^0))^{\otimes 2}$
can be implemented by a Clifford Circuit.

\paragraph*{Extension to whole $M^3$ --- filling up $0$-cells:}

The idea continues from the previous extension.
The commutant of the image of $\unhealedQCA(M^3 \setminus M^0)$
lives on the $0$-skeleton $M^0$,
and there are exactly the same number of anticommuting pairs of Pauli operators in the commutant
as there are unassigned qubits in $M^0$.
The extension is therefore possible, and it is unique up to a QCA on $M^0$, which is a Clifford circuit.

We now have completed the existence proof of $\unhealedQCA(M)$ 
whose action matches that of $\alpha_{WW}(T^3)$ on any cubic lattice patch.

\subsubsection{Disentangling the canonical Hamiltonian of \cref{eq:Hamiltonian-3FWW}}

The QCA $\unhealedQCA(M)$ that is constructed  above (or, rather, is shown to exist),
disentangles the ground state of $H$ in \cref{eq:Hamiltonian-3FWW} at least in each cubic lattice patch.
That is, the group generated by the terms of \cref{eq:Hamiltonian-3FWW} becomes under $\unhealedQCA(M)$
a group that contains all single qubit $Z$-operator on the interior of the cubic lattice patches.
If $\Psi$ is the ground state of $H$ in \cref{eq:Hamiltonian-3FWW},
let $\unhealedQCA(\Psi)$ be the ground state of $\unhealedQCA(M)(H)$.
Any remaining entanglement in $\unhealedQCA(\Psi)$ is in the $2$-skeleton (in any sense to a reader's taste!).

We claim that $\unhealedQCA(\Psi)$ can be mapped to a completely trivial product state by a Clifford circuit of range $O(\smallR)$.
To show this we recall that $\Psi$ does not have any deconfined topological charge.
$\unhealedQCA(M)$ being locality preserving does not alter the absence of topological charges.
Recall that the construction of $\unhealedQCA(M)$ preserves translation invariance whenever possible.
In particular, on any $2$-cell between two cubic lattice patches,
our QCA $\unhealedQCA(M)$ as well as the Hamiltonian \cref{eq:Hamiltonian-3FWW}
is translation invariant within the $2$-cell,
and therefore $\unhealedQCA(M)(H)$ obeys the same property.
Discarding single qubit $Z$ factors in $\unhealedQCA(M)(H)$ that are on the interior of the $3$-cells,
we have a translation invariant commuting Pauli Hamiltonian in the interior of $2$-cells,
without any topological charges.
Note that since $H$ obeys the local topological order condition~\cite{BravyiHastingsMichalakis2010stability},
so does $\unhealedQCA(M)(H)$.
It is proved that for translation invariant 2-dimensional commuting Pauli Hamiltonians with the local topological order condition, 
no topological charge is synonymous to nondegeneracy of the ground state~\cite[\S 7.Thm.4]{Haah2013}.
Then, \cite[Thm.IV.4]{haah2018nontrivial} says any such Hamiltonian is disentanglable by a translation invariant Clifford QCA,
and \cite[Thm.1]{haah2019clifford} says any translation invariant Clifford QCA in 2-dimension is a Clifford circuit followed by a shift.
Therefore, there exists a translation invariant Clifford circuit~$\unhealedQCA_2$ that disentangles 
$\unhealedQCA(\Psi)$ in the interior of all $2$-cells.

The state $\unhealedQCA_2 \circ \unhealedQCA(\Psi)$ has entanglement only in $M^1$,
but by the same argument using the fact that every translation invariant Clifford QCA in 1D 
is a Clifford circuit $\unhealedQCA_1$ followed by a shift~\cite{haah2019clifford},
we disentangle $\unhealedQCA_2 \circ \unhealedQCA(\Psi)$, pushing any remaining entanglement down to $0$-skeleton.
It is trivial to disentangle the $0$-cells by a Clifford circuit $\unhealedQCA_0$.
The composition $\unhealedQCA_0 \circ \unhealedQCA_1 \circ \unhealedQCA_2$ is a desired Clifford circuit.
Our disentangling QCA is finally 
\begin{align}
\alpha_{WW}(M) = \unhealedQCA_0 \circ \unhealedQCA_1 \circ \unhealedQCA_2 \circ \unhealedQCA(M).
\end{align}

\subsection{Phase ambiguity and construction of $\Umat$} \label{sec:phase}

We are now ready to define $\Udis$ and show the properties in \cref{sec:UdisProperties}.

We choose a continuous projection map $\phi$ from the $2$-skeleton of our $4$-manifold to $\mathbb R^2$
such that it is injective on every $2$-cell. 
The existence of such a $\phi$ follows for the same reasons as we explained in \cref{sec:3FWWarb}.
We use the same $\phi$ for all $3$-dimensional submanifolds that are domain walls of spin configurations $\vec z$.
That is, the Hamiltonian terms on a $3$-dimensional domain wall $M$ are determined by $\phi$ restricted to the $2$-skeleton of $M$.
In particular, they depend only on the local information of the spin configuration.
Every edge of the material has two qubits, 
one of which is regarded as placed on a secondary lattice obtained by a generic perturbation of
the triangulation.

The disentangling QCA $\alpha_{WW}(M)(\cdot)$ defined by $\alpha_{WW}(M)(O) = (\Umat_{\vec z})^\dagger O \Umat_{\vec z}$
determines $\Umat_{\vec z}$ up to a phase factor that may depend on $M = M(\vec z)$.
We fix this phase factor by the \emph{positive vacuum} rule
\begin{align}
\bra 0 \Umat_{\vec z} \ket 0 > 0 \text{ for all } \vec z \label{eq:phaserule}
\end{align}
where $\ket 0$ is the vacuum state without any string segment present.
This vacuum state is a single state of our $4$-dimensional lattice
that makes sense for all choices domain walls of spin configurations.
This condition is possible only if $\bra 0 \Umat_{\vec z} \ket 0 \neq 0$,
which asserts that the amplitude of $\ket 0$ in $\Psi(M)$ be nonzero.
We have shown that this is the case when we write \cref{eq:3FWW-GS-on-M} above.
\cref{eq:phaserule} means that we fix the amplitude of $\ket 0$ in $\Psi(M)$ to be positive.

It remains to show the defining properties in \cref{sec:UdisProperties}.
In particular, we have to show the two conditions in \cref{lem:circfromcontrolledqca}.
The first condition that $(\Umat_{\vec z})^\dagger \Umat_{\vec z +i}$ be a local operator around $i$
follows from the construction of $\alpha_{WW}(M)$
where we defined the action of $\alpha_{WW}(M)$ at a simplex $\Delta$ by the open star of $\Delta$
(the collection of all cells that intersect $\Delta$).
Since the action of $\alpha_{WW}(M)$ is defined locally, the action of 
$(\Umat_{\vec z})^\dagger \Umat_{\vec z +i}$ as a QCA far from $i$ is the identity.
The second condition that $(\Umat_{\vec z})^\dagger \Umat_{\vec z +i} = (\Umat_{\vec z + j })^\dagger \Umat_{\vec z +i +j}$ for far separated $i$ and $j$
requires our phase rule in \cref{eq:phaserule}.
Let us rearrange the terms and indices to rewrite the condition as 
\begin{align}
\Umat_{\vec z +i} (\Umat_{\vec z})^\dagger = \Umat_{\vec z +i +j} (\Umat_{\vec z + j })^\dagger.
\end{align}
We know this condition holds up to a phase factor because both sides have the same action as a QCA.
To prove the equality, it suffices to show
\begin{align}
\bra 0 \Umat_{\vec z +i} (\Umat_{\vec z})^\dagger \ket 0  > 0
\end{align}
for any $\vec z$ and $i$.
This is equivalent to saying that the overlap of two Walker-Wang ground states 
with their global phase factors fixed by \cref{eq:phaserule} is also positive.

In fact, we claim an even stronger result
\begin{align}
\braket{\Psi(M)|\Psi(M')} > 0
\end{align}
for two arbitrary $3$-manifolds $M,M'$ embedded in some common space.
If we expand the ground state wavefunction in the string configuration basis $\{ \ket s \}$,
then the overlap is $\sum_s \braket{\Psi(M)|s} \braket{s|\Psi(M')}$.
The string configuration $s$ is a $1$-cycle that is nullhomologous within $M$ as well as within $M'$.
Applying \cref{lem:1-cycle-to-amplitude} to $M \cup M'$ which is a (cubulated) simplicial complex, 
inherited from our $4$-dimensional complex,
we see that $\braket{s|\Psi(M)} = \braket{s|\Psi(M \cup M')} = \braket{s|\Psi(M')}$.
Therefore the overlap is positive.

We have completed the construction of our $4$-dimensional state $\Psi_0$.

\subsection{Crane-Yetter TQFT}
\label{sec:CY}

We remark that there is an alternative way to fix the phase, different from that in \cref{sec:phase}.  Consider two choices of spin configurations,
$\vec z$ and $\vec z+i$ differing by a flip of a single spin $i$.
Each spin defines some boundary of that spin configuration, defining a cellulation of a three-manifold.
These two different cellulations, denoted $c_1,c_2$, agree on some region that we will call the ``common region", namely the region far from spin $i$.
We can use the Crane-Yetter model\cite{crane1997state} to define an operator supported near spin $i$ that maps the ground state of the Walker-Wang
on cellulation $c_1$ to that on $c_2$.  We will denote this operator by $O_{\vec z+i,\vec z}$.

The Crane-Yetter model is a state sum model defined on a cellulation of a four-manifold.  If the four-manifold has a boundary, the amplitude of the state sum depends on the configuration on the boundary, and it reproduces the Walker-Wang ground state wavefunction.  Let $b$ be the boundary of the common region and let $r_1,r_2$ be the regions of $c_1,c_2$ outside the common region.  Define a bordism from $r_1$ to $r_2$ relative to boundary $b$.
Define then a Crane-Yetter state sum from this bordism, taking a constant coloring on the boundary.
This gives the operator $O_{\vec z+i,\vec z}$ above\cite{wang}.

Indeed, we can choose the bordism in an obvious way.
First we construct a bordism from $c_1$ to $c_2$ and then we construct a relative bordism.  To construct the bordism,
consider a $5$-manifold given by the four-dimensional ambient space crossed with an interval $[0,1]$.  Call the coordinate on the interval ``time" $t$.  Define a submanifold $M$ to contain all four-cells with spin down in $\vec z$ for $t<1/2$ and to contain all spins down in $\vec z+i$ for $t\geq 1/2$.  Take the boundary of $M$ and intersect it with the open interval $(0,1)$ in the time coordinate, i.e., remove the components of the boundary at times $t=0$ and $t=1$.  This is the desired bordism as
its boundary is
$c_1$ at $t=0$ and $c_2$ at $t=1$.
Then, since this bordism is simply a product in the common region, it gives the desired relative bordism.

Remark: one may regard this bordism as describing a ``spacetime history" of spins, where the spin configuration changes from $\vec z$ to $\vec z+i$, transitioning abruptly at $t=1/2$, and where the bordism is the boundary of the down spin configuration.

With this choice of bordism, we can then make a choice to fix the phase in $\Umat$.  The choice will be done in terms of relative phases.  Recall how we showed that $\Udis$ can be realized as a quantum circuit by considering a sequence of single spin flips.  We use a similar technique here.  Choose any arbitrary ordering on the spins.  As a first try (we will see that this works if the ambient four-dimensional manifold has even Euler characteristic and we will explain how to modify it later), given a spin configuration $\vec z$,
define $\Umat_{\vec z}$ so that acting on $|0\rangle$ it produces the Walker-Wang ground state with the same phase as is produced by acting with a sequence of the operators $O_{\vec z+i,\vec z}$ defined from the Crane-Yetter model to start with the configuration with all $\vec z_i=+1$ and flip spins in turn (using the ordering above) until arriving at configuration $\vec z$.

Note in fact that the resulting phase of $\Umat_{\vec z}$ is {\it independent} of the arbitrary ordering of spins chosen.

We must verify two properties of this phase rule: that it obeys the ${\mathbb Z}_2$ symmetry and that it has the needed locality properties.  To verify the symmetry, we must compute $\Umat(-\vec z)^{\dagger} \Umat(\vec z)$.  This quantity is a scalar and so can be evaluate on the vacuum $|0\rangle$.  The result may be seen to be the partition of the Crane-Yetter model for a cobordism from the empty manifold to itself, i.e., the partition function of the Crane-Yetter model on some closed four-manifold.
This closed manifold is the boundary of some spacetime history of spins starting with all spins $+1$ and ending with all spin $-1$. and so the resulting four-manifold is cobordant to the ambient four-manifold on which the system is defined.
The partition function for the three-fermion Walker-Wang model is then equal to $-1^{\chi}$, with $\chi$ the Euler of the ambient four-manifold; this follows from the fact that the signature and Euler characteristic have the same parity and from the formula for the partition function in terms of signature\cite{crane1997state}.

If $\chi$ is even, then the $\mathbb{Z}_2$ symmetry is obeyed.  If $\chi$ is odd, then the system is {\it anti-symmetric} under ${\mathbb Z}_2$.  This can be fixed either by conjugation $\Udis$ by a Pauli $Z$ operator on a single spin or, in a less ad hoc way, by combining with the dual model to the generalized double semion model, which has the same sign rule\cite{Freedman_2016,fidkowski2019disentangling,debray2018low}.

We now consider locality.  We need to show that \begin{align}
\label{needshowloc}
(\Umat_{\vec z+i})^\dagger \Umat_{\vec z}=(\Umat_{\vec z+i+j})^\dagger \Umat_{\vec z+j}
\end{align}
 for any spin $j$ sufficiently distant from spin $i$.
Of course, the left and right sides of that expression agree by construction up to a phase; we need to show that the phase is the same.
Both sides of the expression are, by construction, supported near spin $i$.

Now, we use the fact that the phase rule is independent of the arbitrary ordering of spins chosen: any two orderings will give the same
phase since the two orderings will give bordisms that are cobordant to each other.
Hence, we can choose $i$ to be the last in the ordering.
Define $\psi_{\vec z}=\Umat_{\vec z}|0\rangle$.
The operator
$(\Umat_{\vec z+i})^\dagger \Umat_{\vec z}$ has its phase such that
$$\langle \psi_{\vec z} |  O_{\vec z+i,\vec z}^\dagger  (\Umat_{\vec z+i})^\dagger \Umat_{\vec z}   |\psi_{\vec z}\rangle =1.$$
Since $O_{\vec z+i,\vec z}=O_{\vec z+i+j,\vec z+j}$ for $j$ sufficiently separated from $i$ (this holds by construction since the bordism depends only on spins near $i$),
and since the reduced density matrix of $\psi_{\vec z}$ near spin $i$ is the same as the reduced density matrix of $\psi_{\vec z+j}$ near spin $i$, we have have that
\begin{align}
\label{phase1}
\langle \psi_{\vec z+j} |  O_{\vec z+i+j,\vec z+j}^\dagger  \Umat_{\vec z+i})^\dagger \Umat_{\vec z}   |\psi_{\vec z+j}\rangle =1.
\end{align}
Further, by construction
\begin{align}
\label{phase2}
\langle \psi_{\vec z+j} |  O_{\vec z+i+j,\vec z+j}^\dagger  (\Umat_{\vec z+i+j})^\dagger \Umat_{\vec z+j}   |\psi_{\vec z+j}\rangle =1.
\end{align}
Eq.~(\ref{needshowloc}) follows.

\subsection{$\Psi_0 \otimes \Psi_0$ is trivial.}\label{sec:double-triviality}

We show that two copies of our ground state $\Psi_0$ represent the trivial phase under $\ZZ_2$ symmetry.
This will follow from the fact that
$\Udis \otimes \Udis$ can be written as a quantum circuit in which every gate commutes with $\symX$.

First, we claim that for any unitary $U$ that acts as a QCA by conjugation
and any onsite symmetry operator $\mathbf g = \bigotimes_i g$,
if $U \mathbf g = \mathbf g U$, then $U \otimes U^\dagger$ is a quantum circuit of depth $O(1)$
of which every gate commutes with $\mathbf g$.
The proof of this is completely analogous to the proof that $\alpha \otimes \alpha^{-1}$ 
is a quantum circuit for any QCA $\alpha$.
A swap gate is always symmetric under $\mathbf g$,
and so is the swap gate conjugated by $U \otimes I$ since $U$ as a whole is symmetric.
Therefore, $U \otimes U^\dagger = 
\big[ (U \otimes I) (\bigotimes \mathrm{SWAP}) (U \otimes I)^\dagger \big] (\bigotimes \mathrm{SWAP})$
is a circuit that consists of symmetric gates.
Therefore, $\Udis \otimes \Udis$ is equivalent to $\Udis^2 \otimes I$ up to a locally symmetric circuit,
where $\Udis^2 = \sum_{\vec z} \Pi^{spin}_{\vec z} \otimes (\Umat_{\vec z})^2$.

Recall that $\Umat_{\vec z}$ is a QCA whose action is determined locally:
on any $3$-cell the action is determined by two neighboring $4$-cells intersecting along that $3$-cell,
and more generally the action on a $k$-cell 
(where a material cubic lattice patch if $k =3$ or some part thereof if $k < 3$ is supported)
is determined by the spins in the open star of that $k$-cell.
The open star has diameter $O(\largeR)$.
Furthermore, since the square of $\alpha_0 = \alpha_{WW}(T^3)$ is a quantum circuit,
our sewing construction gives a circuit for $C_{\vec z} = (\Umat_{\vec z})^2$.
The gates of $C_{\vec z}$ has range $O(\smallR)$,
and hence the range of $C_{\vec z}$ is also $O(\smallR)$.

We are going to rewrite $C_{\vec z}$
as $C_{\vec z} = C^{(0)}_{\vec z} C^{(1)}_{\vec z} C^{(2)}_{\vec z} C^{(3)}_{\vec z}$
where $C^{(k)}_{\vec z} = C^{(k)}_{-\vec z}$ consists of nonoverlapping unitaries
each of which is supported on a $k$-cell.
The nonoverlapping unitaries are all local (support of size $O(\largeR)$),
so we regard them as gates.
Then, $\Udis^2$ is written as a circuit of depth~$4$ 
whose gates are manifestly symmetric under $\symX$.
This will complete the proof that $\Udis \otimes \Udis$
can be written as a quantum circuit of $\ZZ_2$ symmetric gates.

A gate in $C^{(3)}_{\vec z}$ at a $3$-cell $\Delta^3$ is just a unitary supported inside $\Delta^3$
whose action by conjugation is equal to that of $C_{\vec z}$ for all operators on the interior of $\Delta^3$.
Such a unitary can be chosen by e.g. dropping all the gates of $C_{\vec z}$ in the complement of the interior of $\Delta^3$.
Since $C_{\vec z} = C_{-\vec z}$, we may choose the same gates for $C^{(3)}_{\vec z}$ and $C^{(3)}_{-\vec z}$.
Inductively, for $k < 3$, the gates in $C^{(k)}_{\vec z}$ at $k$-cells are 
chosen such that the product $C^{(k)}_{\vec z} \cdots C^{(3)}_{\vec z}$
acts the same way as $C_{\vec z}$ does for all operators in the complement of $(k-1)$-skeleton.
Note that the gates of $C^{(k)}_{\vec z}$ do overlap with those of $C^{(k+1)}_{\vec z}$
and strictly speaking we should have specified $O(1)$ neighborhoods of cells when we match the actions by conjugation with $C_{\vec z}$,
but we did not for clarity of exposition.
Since $C_{\vec z}$ depends only on the geometry of the domain wall, rather than $\vec z$ itself,
the symmetry $\vec z \leftrightarrow - \vec z$ is retained.

\section{Quantized invariants} \label{sec:invariants}

In this section, we construct quantized invariants to classify the beyond cohomology phase.  
We first define a bulk invariant, which is a property of the ground state: 
roughly speaking, it measures the chirality of a symmetry defect.  
Our arguments for why this bulk invariant is well defined and quantized are not rigorous, 
and rely on certain physical assumptions which we spell out.  
We also define several equivalent boundary invariants (the ``$J$, $L$, $M$'' invariants).
Again we do not give any rigorous proofs, 
but we do expect that the quantities we define will be invariant under certain quantum circuits.

The boundary invariants that we define are all properties of the disentangler $\Udis$.
From $\Udis$ one constructs a boundary symmetry operator~\cite{ElseNayak}.  
For the specific boundary disentangler $\Udis$ constructed in \cref{sec:model}, 
we will see that this boundary symmetry is equivalent to the QCA of \cite{haah2018nontrivial}, 
up to a quantum circuit. 
The $J$ invariant will then be defined to be an equivalence class of 
boundary symmetry operators up to quantum circuits.  
We will show that the $J$ invariant coincides with the bulk invariant, 
showing that in fact the $J$ invariant depends only on the bulk ground state, 
and is independent of the choice of disentangler $\Udis$.

The other two equivalent formulations of the boundary invariant are as follows.  
First, the $L$ invariant is constructed from the boundary symmetry operator, 
and measures the chiral central charge between two boundary domains.  
The $L$ invariant is a number equal to $0$ or $4$ mod $8$, 
and is a boundary version of the bulk invariant, 
roughly because a boundary domain wall is analogous to a bulk symmetry defect.  
The $M$ invariant will similarly measure the chiral central charge between two domains, 
doing it in a different way that is more complicated but potentially more general: 
the $M$ invariant is an invariant of {\it three-dimensional} commuting projector Hamiltonians 
and the difference of this value between two different commuting projector Hamiltonians 
(corresponding to the two different domains of the $L$ invariant) 
will in fact give the $L$ invariant.

\subsection{The bulk invariant} \label{sec:bulkinvariant}

\noindent
We now define the bulk invariant.  In this section all dimensions will be spatial.

Let us take the bulk to be ${\mathbb{R}}^4$, parametrized by $(x,y,z,w)$.
Modify the Hamiltonian to insert a $\ZZ_2$ symmetry defect on the $xy$-plane, i.e., at $z=w=0$.
For example, this can be achieved by minimally coupling to a background lattice $\ZZ_2$ gauge field.
Viewing a spatial $\ZZ_2$ gauge field configuration as a three-manifold, 
we can take one that corresponds to a semi-infinite domain wall at $(x,y,z,0)$ with $z>0$.
This produces the desired defect at $z=w=0$.

We would now like to define a notion of chirality (i.e. chiral central charge) on the defect 
analogous to the usual notion of chirality in $2d$ states of bosons.
However, we cannot do this directly, because the usual definition of chiral central charge in $2d$ 
requires us to expose a $1$d edge and examine the energy current on this edge as a function of temperature.
On a defect it is impossible to expose such a $1d$ edge.
To get around this problem, let us first make some observations:

\begin{itemize}

\item[ 1)] We can ensure that there are no anyons living on the defect, 
by choosing the Hamiltonian terms at the defect core appropriately.  
This should always be possible, at least with ancilla degrees of freedom near the defect: 
if one choice of Hamiltonian has some topological order at the defect core, 
we can simply make a copy of the opposite topological order out of the ancilla degrees of freedom 
and condense the appropriate bound states to get rid of this topological order.
Let us denote by $\ket{\Psi_{\rm{defect}}}$ the ground state wave-function of the system 
with such an anyon-free defect.

\item[ 2)] We assume that there exists a circuit $U_{\rm{dis}}^{\rm{defect}}$, 
acting on the whole system, which disentangles the ground state of the beyond cohomology phase, 
with defect inserted as above, 
into a state which looks like a tensor product away from the defect core.
We do not require $U_{\rm{dis}}^{\rm{defect}}$ to have any particular symmetry properties.%
\footnote{
For our specific exactly solved model, such a $U_{\rm{dis}}^{\rm{defect}}$ can be used, 
via dimensionally reducing the angular coordinate in the $zw$-plane, 
to construct a disentangling circuit for the $3$-fermion Walker Wang model ground state.  
Thus constructing $U_{\rm{dis}}^{\rm{defect}}$ will be at least as difficult as 
constructing a circuit disentangler for the $3$-fermion Walker Wang model, 
and in particular $U_{\rm{dis}}^{\rm{defect}}$ will presumably have tails.
}

\item[ 3)] In a $2d$ bosonic gapped phase with no anyons, 
it is believed~\cite{Kitaev_2005} that chiral central charge 
is quantized in integer multiples of $8$, $c_{-} = 8n$.  
Furthermore, any such phase is believed to be finite depth circuit equivalent to $n$ copies of the $E_8$ state.  
This is usually summarized by saying that there 
is a $\ZZ$ classification of invertible bosonic phases in $2d$.

\end{itemize}

We can now define our bulk invariant as follows.  
First, note that $U_{\rm{dis}}^{\rm{defect}} \ket{\Psi_{\rm{defect}}}$ 
is a product state away from the defect core, and, 
because it does not support anyons, 
it must be equivalent to some number $n_0$ of $E_8$ states at the location of the defect core.
Let us now slowly rotate the defect by an angle $\pi$ in the $yz$-plane.
Specifically, we define a family of defect Hamiltonians $H(\varphi)$, 
parametrized by an angle $\varphi$ varying from $0$ to $\pi$, 
such that $H(\varphi)$ has a defect extending in the ${\hat{x}}$ 
and $\cos (\varphi) {\hat{y}} + \sin (\varphi) {\hat{z}}$ directions 
(the corresponding $\ZZ_2$ gauge field configuration 
can be viewed as the appropriate half of $xyz$-space, with boundary equal to the defect). 
For general $\varphi$, we choose the Hamiltonian terms in the core of the defect to be 
such that there are no phase transitions or level crossings as a function of $\varphi$.
Let $\ket{ \Psi^{\varphi}_{\rm{defect}} }$ be the ground state of $H(\varphi)$.
Note that $H(\pi)$ is gauge-equivalent to $H(0)$: 
conjugating by the unitary $\tilde \symX$ that acts by the $\ZZ_2$ symmetry in the region $z>0$ 
takes $H(\pi)$ to $H(0)$.
Thus $\tilde \symX \ket{ \Psi^{\pi}_{\rm{defect}} }$ looks like $\ket{ \Psi^{0}_{\rm{defect}} }$ 
away from the defect core.  
Hence $U_{\rm{dis}}^{\rm{defect}} \tilde \symX \ket{ \Psi^{\pi}_{\rm{defect}} }$ 
is a product state away from the defect core, 
and thus equivalent to some number $n_1$ of $E_8$ states.  
We claim that the \emph{parity} of $n_0-n_1$ is a bulk invariant which diagnoses the beyond cohomology phase.

One can argue on physical grounds that this bulk invariant is well defined.
First, the number of $E_8$ states is an integer whose sign depends 
on the choice of frame in our $4$-dimensional space,
but given a frame the parity of $n_0 - n_1$ is unaffected.
Second, more importantly, the choice of Hamiltonian at the defect core was arbitrary.
Given the constraint that there are no anyons at the defect core,
the only ambiguity is in the number $n_0$ of $E_8$'s at the defect core, 
after the disentangling circuit $U_{\rm{dis}}^{\rm{defect}}$ is applied.
Now, if we tack on some extra number $m$ of $E_8$'s at the defect core, 
then $n_0 \rightarrow n_0 + m$, but also $n_1 \rightarrow n_1 - m$, 
since these extra $E_8$'s simply get rotated to their inverses during the $\pi$ rotation.
Hence the parity of $n_0 - n_1$ is unaffected.

It remains to show is that the exactly solved model we built in \cref{sec:model} 
has an odd value of $n_0 - n_1$, and is hence nontrivial.
We will do this indirectly, by first defining the boundary $J$ invariant, 
showing that it is equivalent to the bulk invariant, and that it is nontrivial for our model.

\subsection{The $J$ Invariant and the Boundary Symmetry Operator} \label{subsec:J}

We follow the procedure described in Ref.~\onlinecite{ElseNayak} to construct an effective action of the symmetry at the boundary of a four-dimensional system.  Let us review this procedure. 
Given a unitary such as $\Udis$, we will define its {\it restriction} to some large $4$-ball $B$ to be any unitary $\tUdis$ which acts by conjugation as a QCA with range $O(1)$ and which has the same action by conjugation as $\Udis$ does on operators which are supported within $B$ a distance $>O(1)$ from the complement of $B$ (i.e., those supported in the interior of $B$, away from the boundary).  For example, we may choose $\tUdis$ to be $\Udis$ itself, or we may choose $\tUdis$ to be the product of the gates in $\Udis$ which
are supported within that ball $B$.
We will choose the restriction of $\symX$ to be the product of Pauli $X$ operators in the ball, calling the resulting operator $\tilde \symX$.
Let us also define $\tilde{\tilde \symX}$ to be the product of Pauli $X$ operators in a slightly smaller ball $B'$, 
choosing the radius of $B'$ smaller than that of $B$ by an amount $O(1)$ that is still large compared to the range of $\Udis$.
Then we define the boundary symmetry operator $\Xbdry$ by
\begin{align} \label{eq:Xbdry}
\Xbdry=\tilde{\tilde \symX} \tUdis^\dagger \tilde \symX \tUdis.
\end{align}
Since $\tUdis$ has the same action by conjugation on $\tilde \symX$ as $\Udis$ does, and since $\symX \Udis^\dagger \symX \Udis=I$,
the product $\tUdis^\dagger \tilde \symX \tUdis$ is equal to a product of Pauli $X$ operators inside the ball, multiplied by some unitary supported near the boundary $\partial B$.
Choosing the radius of $B'$ appropriately, the Pauli $X$ operators in $B'$ are cancelled by Pauli $X$ operators in $\tUdis^\dagger \tilde \symX \tUdis$ and
so $\Xbdry$ is supported near $\partial B$ and
\begin{align}
\Xbdry^2=I.
\end{align}

Then, we have
\begin{lemma}\label{indeplemma}
The boundary operator $\Xbdry$ is independent of the choice of restriction $\tUdis$, up to multiplying $\Xbdry$ by a quantum circuit supported near $\partial B$.
Hence, up to a quantum circuit supported near $\partial B$, the boundary operator $\Xbdry$ is equal to
 $\tilde{\tilde \symX} \Udis^\dagger \tilde \symX \Udis $, i.e., we choose $\tUdis=\Udis$ here.
\end{lemma}
Remark: if we allowed more general restrictions $\tilde \symX$, 
 then $\Xbdry$ is also
independent of that choice of restriction too, so long as the restriction $\tilde \symX$ is supported on $B$ (and obeys the other requirements of a restriction, i.e., acting as $\symX$ near the interior of $B$ and acting as a QCA).
\begin{proof}
Let $\tUdisz$ be some restriction, with corresponding boundary
operator $\Xbdryz=\tilde{\tilde \symX} \tUdisz^\dagger \tilde \symX \tUdisz$.
Consider some other restriction
$\tUdis=\tUdisz V$.
By definition of a restriction, $V$ must act as the identity on the interior of $B$ and must act as a QCA.
The boundary operator corresponding to $\tUdis$ is
$\Xbdry=\tilde{\tilde \symX} V^\dagger \tUdisz^\dagger \tilde \symX \tUdisz V$.
By assumption on the support of $V$, this is equal to
$V^\dagger \Xbdryz V$.
We may commute $V$ through using the group commutator and the result is equal to $\Xbdryz$ up to a quantum circuit;
this is obvious if $V$ is a circuit and more generally, if $V$ is a QCA, 
we may use the general fact that the commutator of QCA is always a circuit~\cite{FreedmanHaahHastings2019Group}.

The same technique works also if we allow more general restrictions $\tilde \symX$.
\end{proof}

With this lemma in hand, we will choose particular restrictions $\tUdis,\tilde \symX$ that make the computation of $\Xbdry$ simpler.  With these choices, we will show show that $\Xbdry$ acts by conjugation as a QCA and this QCA is equivalent, up to a
quantum circuit, to the QCA $\alpha_{WW}$ that disentangles the three-fermion Walker-Wang phase of Ref.~\onlinecite{haah2018nontrivial}.
We emphasize that this quantum circuit is 
a quantum circuit supported near (i.e., within distance $O(1)$) the three-dimensional boundary; 
this is important because, by a ``swindle,'' 
it is possible to realize any nontrivial three-dimensional QCA on the boundary 
by a quantum circuit in the four dimensional ball.

In words, we will choose $\tUdis$ to be a controlled unitary, treating all spins outside $B$ as if they were in the $Z=+1$ state.
Precisely,
we will choose
\begin{align}
\tUdis = \sum_{\vec z} \Pi^{spin}_{\vec z} \otimes \Umat_{\vec z_B}.
\end{align}
Recall that our notation is that $(\vec z_B)_i=\vec z_i$ for $i\in B$ and $(\vec z_B)_i=1$ for $i\not \in B$.
Note that for any gate in $\Udis$, if all spins in the support of that gate are in the $Z=+1$ state, that gate acts as the identity on the material.  Hence all gates depending only on spins outside $B$ may be dropped from $\tUdis$.

We will consider the action of $\Xbdry$ on the material for various spin basis states.
To get oriented, let us first consider the action of $\Xbdry$ on a configuration with all spins inside the ball in the $Z=+1$ state.  Acting on this state, $\tUdis$ acts as the identity.  Then, flipping spins inside the ball by $\tilde \symX$, the unitary $\tUdis^\dagger$ acts by conjugation on the material as $\alpha_{WW}^{-1}(\partial B)$.
Similarly, if all spins inside the ball are in the $Z=-1$ state, the unitary $\tUdis^\dagger$ acts by conjugation on the material as $\alpha_{WW}(M)$.

Now consider an arbitrary spin configuration inside the ball $B$.  To compute the action of $\Xbdry$ in this case, it is convenient to use the ball $B'$ which is a slightly smaller ball contained within $B$; we choose $B'$ so that the radius of $B'$ is smaller by an amount $O(1)$ that is still large compared to the range of the quantum circuit $\Udis$ and so that $\Xbdry$ acts as the identity on $B'$.  
Let $\tilde \symX'$ be the product of Pauli $X$ operators on spins in $B'$, let $\tUdis'$ be the restriction of $\Udis$ to $B'$ and let $\Xbdry'$ be the boundary symmetry operator
 $\tilde{\tilde {\symX'}} (\tUdis')^\dagger \tilde \symX' \tUdis'$ where $\tilde{\tilde {\symX'}}$ is the product of Pauli $X$ operators in a ball $B''$ slightly smaller than $B'$.

Then, $\Xbdry$ is equal to $\Xbdry'$ up to a quantum circuit near the boundary.  Further,
since $\Xbdry$ commutes with any Pauli $X$ operator on spins in $B'$ and so we may assume that all spins in $B'$ are in the $Z=+1$ state initially when computing $\tilde \symX \tUdis^\dagger \tilde \symX \tUdis$.
However, for this choice of spins in $B'$, we know from above that $\Xbdry'$ acts by conjugation as $\alpha_{WW}^{-1}(\partial B')$.
Hence, for an arbitrary spin configuration inside $B$, we find that $\Xbdry$ acts by conjugation as $\alpha_{WW}^{-1}(\partial B')$, up to a quantum circuit, which in turn
means that 
$\Xbdry$ acts by conjugation as $\alpha_{WW}^{-1}(\partial B)$, up to a quantum circuit.
Note also that $\alpha_{WW}^{-1}(\partial B)$ and $\alpha_{WW}(\partial B)$ agree up to a quantum circuit.

Recall that $\Xbdry^2=I$.  This is interesting because $\alpha_{WW}(\partial B)^2$ is not the identity, but is equal to a circuit and we do not currently know any Clifford QCA that is equal to $\alpha_{WW}$ up to a circuit and whose square is the identity ($\Xbdry$ is not a Clifford QCA).

As we remarked before, since $\Xbdry$ cannot be truncated to act only on part of the boundary, this prevents one from implementing
the second step
of the Else-Nayak descent procedure\cite{ElseNayak}.

We define the $J$ invariant to be the equivalence class of $\Xbdry$ up to quantum circuits.
If $\alpha_{WW}$ is indeed nontrivial as believed, then \cref{Unontrivlemma} implies that $\Udis$ itself is nontrivial.
Here, nontriviality of $\Udis$ means that $\Udis$ cannot be written as a product of local gates which commute with $\symX$.  
\begin{lemma}
\label{Unontrivlemma}
If $\Udis$ can be written as a quantum circuit whose gates commute with $\symX$, then $\Xbdry$ can be written as a quantum circuit supported near the boundary.
\begin{proof}
Suppose $\Udis$ can be written as a product of local quantum gates which commute with $\symX$.
Insert this decomposition of $\Udis$ into
 $\tilde \symX \Udis^\dagger \tilde \symX \Udis $, which is equal to $\Xbdry$ up to a quantum circuit as shown in \cref{indeplemma}.
Then, all gates (both inside and outside $B$) which
are sufficiently far from $\partial B$ cancel, so $\Xbdry$ is a quantum circuit supported near the boundary.
\end{proof}
\end{lemma}

\subsection{Dependence of boundary symmetry on choice of cellulation}

In the previous subsection, we showed that $\Xbdry$ is equal to $\alpha_{WW}(\partial B)$, up to a quantum circuit.  In Ref.~\onlinecite{haah2018nontrivial}, we presented strong evidence that the QCA $\alpha_{WW}$ which disentangles the Walker-Wang Hamiltonian on a cubic lattice is nontrivial.
In this subsection we discuss the relation between these two QCA, and sketch a proof that they agree up to a quantum circuit.

Our method will be, roughly, to show that given a QCA which disentangles the Hamiltonian on some given cellulation, it agrees (up to a circuit) with a QCA which disentangles the Hamiltonian on a refinement of that cellulation; then we show that any two QCA which disentangle the same Hamiltonian agree up to a circuit.  In this way, we compare QCA on two different cellulations by considering a common refinement.

First, let us define precisely what it means for a QCA to disentangle a Hamitonian which is  sum of commuting local terms.  We will say that a QCA $\alpha$ disentangles some Hamiltonian which is a sum of commuting local terms if every term is mapped by $\alpha$ to some product of Pauli $Z$ operators and if every Pauli $Z$ operator is mapped by $\alpha^{-1}$ to a product of terms in the Hamiltonian, with every term in the product supported near the given Pauli $Z$ operator.  Note that this implies that each term in the Hamiltonian squares to the identity.

\paragraph{Any QCA $\gamma$ which maps every Pauli $Z$ operator to itself is a quantum circuit.}
We call such QCA {\it diagonal QCA} because they are given by conjugation by some diagonal unitary $U$, i.e., it is a phase controlled by the qubits.  To show that the QCA is a quantum circuit, tile the manifold with colored tiles as in \cref{lem:circfromcontrolledqca}.  Define $\vec z_{S_a}$ as in that lemma and define $U_{S_a}$ to be the phase corresponding to
qubit configuration $\vec z_{S_a}$. Remark: here, of course, we do not have a distinction between ``spins" and ``material" as we only have ``material" and the
vector $\vec z_{S_a}$ now refers to the configuration of all qubits.
Then, as in that lemma, it is easy to show that $U_{S_{a+1}} U_{S_a}^\dagger$ is a circuit and composing these circuits give the result.

\paragraph{Any two QCA $\alpha,\beta$ which disentangle the same Hamiltonian agree up to a quantum circuit.}
The product $\delta\equiv \alpha \circ \beta^{-1}$ maps every Pauli $Z$ operator to a product of Pauli $Z$ operators.  
We identify a product of Pauli $Z$ operators with a vector in ${\mathbb F}_2^N$ in the obvious way, where $N$ is the number of qubits.  Consider any qubit $i$.
By linear algebra, there is some $j$ such that 
$Z_j$ is linearly independent (using this identification of operators with vectors) of the set of $\delta(Z_k)$ for $k\neq i$.
Hence, $\delta^{-1}(Z_j)$ is linearly independent of the set of $Z_k$ for $k\neq i$ and hence $\delta^{-1}(Z_j)$ has $Z_i$ as a
nontrivial factor.  One may then precompose $\delta$ with CNOTs whose target qubit is $i$ and whose source qubits are the other factors of $\delta^{-1}(Z_j)$ to give a new QCA that maps $Z_i$ to $Z_j$
and that
maps $Z_k$ for $k\neq i$ in the same way as $\delta$.
The range of these CNOT gates is bounded by the range of $\delta$.
We can then 
apply this operation in parallel to
 a set of far separated qubits $i_1,i_2,\ldots$, i.e., we can apply a sequence of CNOTs in parallel, giving a quantum circuit whose gates are supported near $i_1,i_2,\ldots$ so that the composition of this sequence of CNOTs with $\delta$ maps $Z_{i_a}$ to $Z_{j_a}$.
Doing this over a bounded number of rounds maps $\delta$ to a diagonal QCA which, by above, is a quantum circuit.

We now compare the QCA which disentangle the Walker-Wang Hamiltonian $H(c,\phi)$ defined in \cref{eq:Hamiltonian-3FWW}
for different cellulations $c$ and projections $\phi$.
Throughout, when we refer to a ``cellulation," we mean both a cellulation and an associated projection;
for simplicity we omit $\phi$ and write $H(c)$.
We will call a QCA a {\it Clifford QCA} if it maps Pauli operators to products of a Pauli operators.

\paragraph{Let $c_1,c_2$ be cellulations (possibly of different three-manifolds) 
which agree except for a ball~$B$ of radius~$O(1)$.
Assume $\alpha$ is a Clifford QCA that disentangles $H(c_1)$.
Then, there is a Clifford QCA $\beta$ which disentangles $H(c_2)$ such that $\alpha,\beta$ agree up to
a Clifford unitary supported near $B$.}
Before giving the proof, let us remark that of course not every Hamiltonian can be disentangled by a QCA;
for example a Hamiltonian with intrinsic topological order cannot be disentangled.
Our construction of a QCA $\beta$ to disentangle $H(c_2)$ will use the existence of $\alpha$ 
which disentangles all terms of $H(c_2)$ supported in the common region (the complement of $B$),
as well as the fact that $H(c_2)$ is a stabilizer Hamiltonian without any local observable on the ground space.
The last property is trivial once we know $H(c_2)$ has a unique ground state,
but, strictly speaking, we only know that $H(c)$ obeys local topological order without deconfined topological charges.
In fact, our construction of $\alpha_{WW}(M)$ on a cubulated 3-manifold $M$ 
is a reason for, rather than a consequence of, the nondegeneracy of $H(M)$'s ground state.
{\it Proof of the claim {\bf c}}:
Let ${\cal A}$ be the set of Pauli $Z$ operators supported sufficiently far (distance $>O(1)$) from $B$.
Note that each such Pauli $Z$ operator is in the stabilizer group generated 
by $\alpha$ applied to the terms supported in the common region.
Since $H(c_2)$ is a stabilizer Hamiltonian, by linear algebra over $\mathbb F_2^N$, 
we may find some maximal subset ${\cal S}$ of the 
terms of $H(c_2)$ which are linearly independent from each other and from $\alpha^{-1}({\cal A})$, 
i.e., ${\cal S}$ and $\alpha^{-1}({\cal A})$ together give a linearly independent set 
which generates all terms of $H(c_2)$.
The terms in ${\cal S}$ are all supported near $B$ as every term supported in the common region 
is mapped by $\alpha$ to a product of Pauli $Z$ operators 
and if the term is sufficiently far from the complement of the common region,
that product will be in ${\cal A}$.
The terms in ${\cal S}$ are mapped by $\alpha$ in some arbitrary way to products of Paulis.
By linear independence,
we may apply a Clifford supported near the complement of the common region to map each such product to a Pauli $Z$ operator while leaving ${\cal A}$ invariant.
Since any local operator that commutes with all terms of $H(c_2)$ is a product of terms of $H(c_2)$,
there are ``enough" such terms so that every Pauli $Z$ operator is
in the image of some term in ${\cal S}$ or is in ${\cal A}$.

Remark: possibly cellulation $c_2$ may have a different number of cells from that of $c_1$ outside the common region, for example if $c_2$ is a refinement of $c_1$; in this case, the notion of agreeing up to a quantum circuit is a stable equivalence, i.e., we may tensor $\alpha,\beta$ with the identity QCA on any added cells.

\paragraph{Let $c_1,c_2$ be some cellulations of a three-manifold which agree in some some common region, and disagree in some regions $R_1,R_2,R_3,\ldots$ with the $R_a$ far separated from each other and each $R_a$ of diameter $O(1)$.  Let Clifford $\alpha$ disentangle $H(c_1)$.
Then, there is some Clifford QCA $\beta$ which agrees with $\alpha$ up to a quantum circuit of depth $1$ and which disentangles $H(c_2)$.}
To see this, simply implement the prescription of {\bf c} in parallel for each $R_a$.  Since each $R_a$ has diameter $O(1)$, the unitary supported near $R_a$ is a gate with
diameter $O(1)$.

\paragraph{Given any two cellulations $c,c'$ 
which can be related by a sequence $c,c_1,c_2,\ldots,c'$ 
such that each neighboring pair in the sequence obeys the requirements of {\bf d}, 
then if Clifford QCA $\alpha$ disentangles $H(c)$, 
there is some Clifford QCA $\beta$ which disentangles $H(c')$, where $\alpha,\beta$ agree up to
a quantum circuit of depth proportional to the length of the sequence.}
In this way we relate disentanglers for different cellulations.

Note further that any two triangulations of the same $3$-manifold 
can be perturbed to be transversal to each other and then there exists a common refinement.
If one cellulation $c'$ is a refinement of another cellulation $c$, 
then the results above give a quantum circuit of depth $O(1)$ 
mapping the disentangling QCA for one cellulation to that for the other.
The range of gates in the QCA is comparable to the length scale of the cells in $c$, 
which may potentially be much larger than the cells in $c'$.

Finally, putting all this together: 
given two cellulations $c,c'$ of the same manifold, with Hamiltonians $H(c),H(c')$ 
and given two Clifford QCA $\alpha(c)$,$\alpha(c')$ which disentangle the respective Hamiltonians, 
we claim that $\alpha(c),\alpha(c')$ agree up to a local quantum circuit.
Proof: find some common refinement $c''$.
By {\bf e}, $\alpha(c)$ agrees, up to local quantum circuit, with some QCA $\gamma$ that disentangles $H(c'')$.
Similarly, $\alpha(c')$ agrees, up to local quantum circuit, with some QCA $\delta$ that also disentangles $H(c'')$.
Then, by {\bf b}, $\gamma,\delta$ agree with each other up to local quantum circuit.

It is worth remarking that any two QCA (regardless of whether they disentangle the Walker-Wang Hamiltonian) 
which agree on three-cells must agree up to circuits composed with shifts.
This is because of the general result that two-dimensional QCA are trivial~\cite{FreedmanHastings}.

\subsection{Equivalence of $J$ invariant and bulk invariant}

We now give a physical argument that the $J$ invariant is equal to the bulk invariant.
This in particular implies that the $J$ invariant 
is independent of the choice of symmetric disentangler $\Udis$.
In this section $\Udis$ will refer to a general symmetric disentangler,
in some $\ZZ_2$ symmetric phase whose beyond cohomology SPT order we are trying to diagnose.

The idea is to consider a bulk symmetry defect, as in \cref{sec:bulkinvariant}.
This is a codimension~$2$ surface, attached to a $3$-dimensional ``branch cut.''
The branch cut is just the chosen background $\ZZ_2$ gauge field configuration.
Applying $\Udis$ disentangles the Hamiltonian of the system with the defect 
everywhere except at the branch cut.
In other words, the remaining, nondisentangled terms in the Hamiltonian are localized near the branch cut, 
and can be thought of as a $3d$ Hamiltonian.
For the exactly solved model constructed in \cref{sec:model},
this $3d$ Hamiltonian is simply the 3-fermion Walker-Wang model,
as expected from the decorated domain wall picture.
For a general model with symmetric disentangler $\Udis$, 
it is easy to see that the $3d$ Hamiltonian can be obtained 
by conjugating a trivial Hamiltonian $H_{\rm{trivial}}$ 
(namely, minus the sum of $X$ operators on all spins, 
plus terms that put the ``material'' into a trivial state) 
by a unitary $Y = \tilde{\symX} \Udis \tilde{\symX} \Udis^{-1}$, 
where $\tilde{\symX}$ is the global action of the $\ZZ_2$ symmetry on half of the $4d$ space.  
More specifically, if the branch cut is at $w=0$ with $z>0$ as in \cref{sec:bulkinvariant},
then $\tilde{\symX}$ is the product of $X$ operators on all spins with $w>0$.
Note that although $Y$ acts nontrivially everywhere in the neighborhood of $w=0$, 
we only use it to conjugate the terms of $H_{\rm{trivial}}$ near the branch cut, i.e. for $z>0$.
Note also that if we had exposed a physical boundary at $w=0$ then $Y$ would just be $\Xbdry$ acting on that boundary, see \cref{eq:Xbdry}.
To show that the bulk invariant matches the $J$ invariant, 
we then just have to show that the bulk invariant is nontrivial 
precisely when the QCA $Y$ is in the equivalence class of $\alpha_{WW}$.

To do this, let us first construct $\Udis^{\rm{defect}}$, 
a circuit that disentangles the ground state with the defect (but {\emph{not}} necessarily the Hamiltonian).  
Such a circuit will be given by first acting with $\Udis$, 
and then disentangling the remaining state at $w=0, z>0$ with some circuit $W$.
In the case of our model, this remaining state is the ground state of a Walker-Wang model, 
and so here we are making an assumption: 
that the ground state of $3$-fermion Walker-Wang model has a circuit disentangler,
albeit likely one with tails (i.e., not a strictly short ranged one).
We then define $\Udis^{\rm{defect}} = W \Udis$.
We can think of $\left(\Udis^{\rm{defect}}\right)^{-1} \Udis$ as a ``flux-insertion'' operator:
it maps the ground state without a defect to a ground state with a defect.

The key physical insight now is that when the phase is nontrivial, 
the circuit $W$ ``pumps chirality'' $c=\pm 4$, 
from out at infinity, along the branch cut, to the location of the defect.
This comes from the following physical heuristic used in constructing the 3-fermion Walker-Wang disentangler:
one nucleates bubbles of vacuum inside the 3-fermion Walker-Wang model, 
and percolates them until they eat the whole Walker-Wang ground state.  
The important point is that the naive commuting projector bubble surface 
would host the $3$-fermion topological order, 
and we have to get rid of this topological order before we can percolate the bubbles, 
so as to avoid topological ground state degeneracies on the high genus surfaces that form.
We do this by condensing the topological order 
with a truly $2d$, and hence chiral, realization of the opposite topological order.
This is where we have to make a choice of~$c=\pm 4$, 
and this is the sense in which such a process pumps chirality.

Now, when we do the adiabatic rotation by $\pi$, 
the direction from which the chirality~$c_{-}=4$ is pumped rotates by~$\pi$ as well:
it now comes from $z<0$ rather than from $z>0$.
So instead of getting $c_{-}=4$ we get~$c_{-}=-4$, 
which is a difference of an odd multiple of~$8$.
Thus the bulk invariant, given by the parity of one eighth of this difference, 
is nontrivial precisely when the QCA~$Y$ (and hence also $\Xbdry$)
is in the equivalence class of~$\alpha_{WW}$.
Therefore the $J$ invariant matches the bulk invariant.

\subsection{The $L$ Invariant} \label{subsec:l}

The $J$ invariant is simple to formulate, but may be difficult to diagnose, 
because to diagnose it, we need to construct $\Xbdry$ and 
determine whether or not it is a nontrivial QCA.  
In this section we define an equivalent invariant 
--- the $L$ invariant --- which has a more straightforward physical interpretation.

Let us take the bulk to be parametrized by $(\theta_1,\theta_2,\theta_3,w)$, 
where $w \geq 0$ is the bulk coordinate, with the boundary at $w=0$.  
We have chosen the boundary directions to be parametrized by angles $0 < \theta_j \leq 2\pi$; 
it should be understood that the boundary directions 
are much larger than any microscopic length scale.
Thus the geometry is $T^3 \times {\mathbb{R}}^{\geq 0}$.
We now `double' the system by making it twice as long in the $\theta_1$ direction, 
i.e., the boundary now has size $4\pi \times 2\pi \times 2\pi$.  
The local Hilbert spaces of the new doubled system 
are obtained from the original using the natural $2$-to-$1$ covering map.
From now on we work exclusively with the boundary of the doubled system.
The advantage of the doubled system is that it has a discrete translation symmetry 
$T_1: \theta_1 \rightarrow \theta_1+2\pi$.

We now consider a configuration of spins on the boundary 
with $Z=1$ for $0 < \theta_1 \leq 2\pi$ and $Z=-1$ for $2\pi < \theta_1 \leq 4\pi$.
This configuration is invariant 
under translating by $2\pi$ in the $\theta_1$-direction and flipping all the spins.
Let us define the corresponding symmetry operator $\symX' = \Xbdry T_1$.
We now introduce a Hamiltonian in the region $\epsilon < \theta_1 < 2\pi-\epsilon$
that puts the material in that region into a trivial product state,
where $\epsilon \ll 2\pi$ is still large compared to any microscopic scale, 
and we gap out the material in the $2\pi+\epsilon < \theta_1 < 4\pi - \epsilon$ region 
in such a way that the symmetry $\symX'$ is preserved.  
In our exactly solved decorated domain wall model, 
this amounts to putting the material in the $2\pi+\epsilon < \theta_1 < 4\pi - \epsilon$ region 
into a $3$-fermion Walker-Wang state.

We now pick a Hamiltonian in the region $-\epsilon \leq \theta_1 \leq \epsilon$ 
that gaps out the material there without allowing any anyon excitations.
We then conjugate this Hamiltonian by $\symX'$ 
to obtain a corresponding Hamiltonian in the region $2\pi-\epsilon < \theta_1 < 2\pi+\epsilon$.
The result is a fully gapped boundary.
We now dimensionally reduce this boundary in the $\theta_1$-direction 
to obtain a quasi $2d$ system in the $\theta_2,\theta_3$ directions, 
and measure the chiral central charge of this quasi $2d$ system.
Since there are no anyons, this chiral central charge must be equal to $8n$, where $n$ is some integer.
We define the $L$ invariant to be the parity of $n$.  Using standard arguments we know that the $L$ invariant is a well defined quantized invariant, 
independent of the various arbitrary choices made above.

Informally, the $L$ invariant is nontrivial 
precisely when a boundary domain wall has chiral central charge~$4$~mod~$8$.
The reason for the above `doubled' construction is that 
the chiral central charge of a single domain wall cannot be measured,
because a single domain wall is not a truly $2d$ system.

To see that the $L$ invariant is equivalent to the $J$ invariant, 
we first consider a situation where the $J$ invariant is trivial, 
i.e., $\Xbdry$ is a finite depth quantum circuit. 
Then $\Xbdry$ can be truncated, which means that the region $-\epsilon \leq \theta_1 \leq \epsilon$ can be gapped out with commuting projectors (by conjugating the trivial Hamiltonian by a truncated version of $\Xbdry$ that acts only in the region $\theta_1<0$).  The Hamiltonian in the region $2\pi-\epsilon < \theta_1 < 2\pi+\epsilon$ is then also made out of commuting projectors, which means that the entire quasi $2d$ system is a commuting projector model, and hence $n=0$.

On the other hand, let us consider a situation where the $J$ invariant is nontrivial.  
Specifically, let us look at our decorated domain wall model.  
A naive commuting projector choice of Hamiltonian 
will leave the region $-\epsilon \leq \theta_1 \leq \epsilon$ with the $3$-fermion topological order, 
which can be gotten rid of by introducing additional ancilla degrees of freedom in that region, 
putting them into a chiral quasi-$2d$ $3$-fermion topological order (with $c=4$),
and condensing appropriate bound states.
The key point is that we also have to do this in the region 
$2\pi-\epsilon < \theta_1 < 2\pi+\epsilon$ in a way that respects the $\symX'$ symmetry, 
resulting in another $c=4$ ancilla state in that region.%
\footnote{
$\symX'$ cannot flip the chirality of the ancilla state from $c=4$ to $c=-4$.
If it could, then we could consider a situation 
where the region $-\epsilon \leq \theta_1 \leq \epsilon$ is gapped out without anyons, 
and the region $2\pi-\epsilon < \theta_1 < 2\pi+\epsilon$ is gapped out with commuting projectors 
realizing the $3$-fermion topological order.
Then $\symX'$ would be a locality-preserving unitary that 
when applied to this quasi $2d$ system would change chirality from $c=4$ to $c=-4$.
We could then stack this quasi $2d$ system with yet another chiral $c=-4$ realization of the $3$-fermion topological order; this stacked system has zero net chirality and is adiabatically connected to a commuting projector Hamiltonian, yet under $\symX'$ maps to a system with $c=-8$, which is a contradiction.
}
The total chirality is then $c=8$, and the $L$ invariant is nontrivial.
Hence the $L$ invariant is equivalent to the $J$ invariant.

\subsection{The $M$ Invariant} \label{subsec:m}
The $L$ invariant was a way to calculate the chiral central charge of a two-dimensional interface between two three-dimensional boundary domains, by employing a trick where two identical such interfaces are prepared, related by a translation symmetry.  It is natural to wonder whether this trick was necessary: could one instead directly compute the chiral central charge of a single interface?

Let us more generally try to calculate the chiral central charge of a two-dimensional interface between two different commuting projector Hamiltonians in three dimensions.
So, we consider a system on a three-torus.  As in the previous subsection, use angular coordinates to $\theta_1,\theta_2,\theta_3$ to parametrize the three-torus.  Consider some lattice system, with lattice spacing $\ll 1$ and use some commuting projector Hamiltonian $H_1$ on the sites with coordinates $0+\epsilon \leq \theta_3 \leq \pi-\epsilon$ and use some other 
commuting projector Hamiltonian $H_2$ on the sites $\pi+\epsilon \leq \theta_3 \leq 2\pi-\epsilon$.
Here we assume that $H_1$ and $H_2$
are both taken to be translationally invariant Hamiltonians so that some single rule for the Hamiltonian on a unit cell will specify the Hamiltonian on the each of these two sets of sites, which we will call the three-dimensional bulk regions.
(Here `bulk' does not refer to the $4+1$ dimensional bulk;
this entire subsection deals with a $3+1$ dimensional system 
which is ultimately taken to be the $3+1$ dimensional boundary of our SPT.)

Then, there are two interfaces, one near $\theta_0=0$ and one near $\theta_3=\pi$.  Let there be some arbitrary (not necessarily commuting projector) way to gap these two interfaces using local Hamiltonians.  In distinction to the case of the $L$ invariant, the two ways of gapping the interfaces are unrelated: one may make an arbitrary choice at each interface.  However, we again require that the interface be gapped without creating any anyons.

We would like to measure the chiral central charge of each interface separately.  Unfortunately, if we simply dimensional reduce by ignoring the $\theta_3$ coordinate, all we will compute is the {\it total} chiral central charge of both interfaces.  Assuming that the fundamental degrees of freedom are bosonic, rather than fermionic, this number will be equal to $0$ mod $8$ and indeed it may be an arbitrary such number \cite{Kitaev_2005}.

However, there is a way to measure the chiral central charge of each interface separately as we explain in the next two subsubsections.
\subsubsection{Chiral Central Charge in the Hamiltonian Formalism}
As a starting point, we review the formalism for defining the chiral central charge in the Hamiltonian formalism developed in Ref.~\onlinecite{Kitaev_2005}.
This subsubsection repeats ideas developed there; the point is to have some of the expressions in a form in which it will be apparent how to treat the $3$ dimensional bulk between interfaces which we will consider in the next subsubsection.

The formalism begins by defining an edge current.  In this subsubsection we consider a purely two dimensional system.
Consider a lattice Hamiltonian
$H^{(\infty)}=\sum_j H_j^{(\infty)}$ with ground state $|\Psi\rangle$ and gapped excitations, where the $H_j^{(\infty)}$ are local terms:
each term $H_j^{(\infty)}$ is supported near some lattice site $j$.
Let $$H=\sum_j H_j$$ with $H_j=\beta_j H_j^{(\infty)}$.  
To simulate an edge along the $x$ axis in an infinite geometry, the scalar $\beta_j$ is chosen to be zero for positive $y$-coordinate
and positive for negative $y$ coordinate, increasing to infinity as $y\rightarrow \infty$.
Below we consider instead a disc geometry of large radius $R$; here the scalar $\beta_j$ is zero for radial coordinate of $j$ larger than $R$, and becomes very large for radial coordinate much smaller than $R$.

We consider a system at nonzero temperature $T$.
The chiral central charge is then defined to be $\frac{12}{\pi} T^{-2}$ times the energy current in the negative $x$-direction.  One may fix $T=1$ and multiply $H$ by a scalar instead.
To compute this energy current, one replaces the terms $H_j$ by a new set of terms $\tilde H_j$ so that $H=\sum_j \tilde H_j$ and so that $\tilde H_j |\Psi\rangle=0$.  This replacement removes any bulk energy current.  For an arbitrary gapped local Hamiltonian this can be done while keeping the terms $\tilde H_j$ local in space: they may decay superpolynomially fast \cite{Kitaev_2005,hastings2006solving}.

Then, the current from $k$ to $j$ may be defined by
$$f_{jk}=-i \langle  [H_j,H_k] \rangle,$$
where $\langle \ldots \rangle$ denotes a thermal average at temperature $T=1$.
Using a disc geometry and dividing the disc into three sectors $A,B,C$ (depending on whether the angular coordinate is in the interval $[0,2\pi/3),[2\pi/3,4\pi/3)$ or $[4\pi/3,2\pi)$ respectively),
the total edge current can be defined as
$\sum_{k\in A} \sum_{l \in B} f_{k,l}.$

This formalism using the edge current can be used (as we will see in the next subsection) for the case that we are interested in, namely computing the chiral central charge of two two-dimensional interfaces separately.
However, the edge current has the disadvantage that it requires introducing an edge.  To avoid this, one may (again following Ref.~\onlinecite{Kitaev_2005}) define a bulk current.

There, one introduces a ``two-current" $h_{jkl}$ which is an anti-symmetric function of $j,k,l$, decaying rapidly as the distance increases between any two sites.
This two current $h_{jkl}$ obeys
$$f_{kl}=\sum_j h_{jkl}.$$
Then in the disc geometry above, we have $\sum_{k\in A} \sum_{l \in B} f_{k,l}=\sum_{k\in A} \sum_{l\in B}\sum_{j} h_{jkl}=-\sum_{j\in A}\sum_{k\in B} \sum_{l\in C} h_{jkl}$ where the last equality uses that $\sum_{k\in A} \sum_{l\in B}\sum_{j\in A} h_{jkl}=\sum_{k\in A} \sum_{l\in B}\sum_{j\in B} h_{jkl}=0$ by anti-symmetry.
So,
$$c=-\frac{12}{\pi}\sum_{{j\in A},{k\in B},{l\in C}} h_{jkl}.$$
The dominant terms in this expression for $c$ are those near the origin, the triple contact point of the regions.

This then allows us to define an expression for the chiral central charge on a torus geometry.  One may divide the torus into three regions in some arbitrary way, so long as there at least one triple contact point (indeed, one cannot have exactly one such point).
For example, one may parameterize the torus by a square $[-1,+1] \times [-1,+1]$ with opposite edge identified; then, introduce radial and angular coordinates on the square and divide into three regions depending on the angular coordinate.
Then, pick any such triple contact point $p$ and compute the sum restricted to $j,k,l$ near that point:
\begin{align}
\label{csum}
c=-\frac{12}{\pi}\sum_{j\in A, {\rm near} \, p}\;\sum_{k\in B, {\rm near} \, p}\;\sum_{l\in C, {\rm near} \, p} h_{jkl}.
\end{align}

The two-current $h$ is defined by introducing a {\it path} of Hamiltonians $H(\beta)$, with $H(0)=0$ and $H(\beta)\approx \beta H^{(\infty)}$ for large $\beta$.  
One defines a two-current $g$ with $g={\rm d}h$ where the differential is along this path.
To define the two current, let $\langle A(\tau) B(0)\rangle$ denote the thermal average of $\exp(-(1-\tau) H(\beta)) A \exp(-\tau H(\beta)) B$,
and let $\langle \langle A(\tau) B(0)\rangle\rangle$ denote the connected correlation function
$\langle A(\tau) B(0)\rangle-\langle A \rangle \langle B \rangle$.
We emphasize that this thermal average is computed using the Hamiltonian $H(\beta)$ to define the thermal state.

Define $$\mu(A,B,C)=i\int_0^1 \langle \langle A(\tau) [B,C](0) \rangle\rangle,$$
and finally define
$$g_{jkl}=\mu({\rm d}H_j,H_k,H_l)+\mu({\rm d}H_k,H_l,H_j)+\mu({\rm d}H_l,H_j,H_k).$$
Integrating this expression for $g$ over a path $\beta$ from $0$ to $\infty$ gives an expression for $h$.  It is instructive to verify that
the expression is invariant under rescaling $H^{(\infty)}$ by a positive scalar.

Hence, one may insert this into Eq.~(\ref{csum}) obtaining
\begin{align}
\label{csumint}
c=-\frac{12}{\pi}\int \sum_{j\in A, {\rm near} \, p}\;\sum_{k\in B, {\rm near} \, p}\;\sum_{l\in C, {\rm near} \, p}g_{jkl}.
\end{align}
One assumes that this path can be chosen so that no phase transition occurs along the path; in particular, one wants all correlation functions to be local along the path so that the expression is indeed dominated by terms near the chosen triple contact point.

\subsubsection{Chiral Central Charge of Each Interface}
We now consider a three-dimensional system on a three-torus with two two-dimensional interfaces between the two three dimensional bulk regions.
Each of these bulk regions is a two-torus crossed with an interval, and there is a local commuting projector Hamiltonian in each bulk region.
We write the three-torus as a two-torus crossed with a circle, and we decompose the two-torus into three regions $A,B,C$ with at least one triple contact point between the regions as in the above subsubsection.

We begin with the seemingly more complicated expression for the chiral central charge in terms of the two-current $g$.
Then, ignoring the third coordinate
$\theta_3$ of the torus, the chiral central charge of the resulting {\it two-dimensional system}
can be computing from Eq.~(\ref{csumint}).
The quantity $g_{jkl}$ is equal to $\mu({\rm d}H_j,H_k,H_l)$ plus cyclic permutations.
We see that the expression for $\mu({\rm d}H_j,H_k,H_l)$ vanishes unless both $k$ and $l$ are in an interface, and indeed both must be in the same interface, as otherwise $[H_k,H_l]=0$ since we use a commuting projector Hamiltonian in the bulk regions.

Now we make an assumption: we assume again that we can choose the path $H(\beta)$ so that no phase transition occurs.
As explained in Ref.~\onlinecite{Kitaev_2005}, while we can avoid ordinary symmetry breaking phase transitions by an appropriate path, one might worry about phase transitions between topologically ordered states which cannot be avoided.  While this is not a problem in two dimensions, it may be a problem in three dimensions.  However, our interest here is between two three-dimensional bulk regions without anyons (such as a trivial state and a three-fermion Walker-Wang state), so in this case such a choice may be made.  Indeed, Ref.~\onlinecite{haah2018nontrivial} gives a set of generators for the three-fermion Walker-Wang stabilizer group without any redundancies in this set, so that no phase transition occurs in this case.

Under this assumption that no phase transition occurs, the connected correlation function 
$\langle \langle H_j(\tau) [H_k,H_l](0) \rangle \rangle$ decays exponentially in the distance from $j$  to $k,l$.  Indeed, the decay of such correlation functions is taken as the definition of the absence of a phase transition.

So, the expression in Eq.~(\ref{csumint}) vanishes unless $j,k,l$ are near an interface.  Hence we can divide the expression into two distinct sums, one near each interface.
We define the central charge near a given interface then to be one of these sums.  Calling the two interfaces ``top" and ``bottom" and letting $p_{top}$ denote a triple intersection point in the top interface, we define
\begin{align}
c_{top}=-\frac{12}{\pi}\int \sum_{j\in A, {\rm near} \, p_{top}}\;\sum_{k\in B, {\rm near} \, p_{top}}\;\sum_{l\in C, {\rm near} \, p_{top}}g_{jkl}.
\end{align}

\subsubsection{$M$ Invariant from Interfacial Chiral Central Charge}
This chiral central charge then gives some number for each interface between two three-dimensional commuting projector Hamiltonians.
Let us consider some properties of this number.

In the case of an interface between a three-fermion Walker-Wang model and a trivial Hamiltonian (i.e., a sum of Pauli $Z$ terms), this number is equal to $4$ mod $8$.
However, in other cases this number may be different.
For example, the interface between a trivial model and a Walker-Wang model based on the bosonic $2/3$ state has chiral central charge equal to $2$ mod $8$; see Ref.~\onlinecite{haah2019clifford}.
Conversely, the interface between this Walker Wang model and the trivial model has central charge $-2$ mod $8$.
That is, the sign of the central charge of the interface changes if one reflects the model across the coordinate perpendicular to the interface.\footnote{The reader may wonder how it is that the chiral central charge of various interfaces is already well-known even though we define it here.  The point is that here we give a definition in terms of lattice Hamiltonians without relying on known results from conformal field theory or topological quantum field theory.}

We remark that in the case of the the three-fermion Walker-Wang model, another way to compute the chiral central charge of a $2d$ boundary between this Walker-Wang model and the vacuum would be to gap one such $2d$ boundary without anyons (but using terms that do not commute) and gap the other $2d$ boundary using commuting terms (at the cost of anyons).  Then, ignore the third coordinate and compute the chiral central charge of the resulting two-dimensional system.  

Proceeding with the definition of the $M$ invariant, consider now two Hamiltonians $H_1,H_2$.
Assume that the interface from $H_1$ to $H_2$ can have central charge $c$.  Then, the interface from $H_2$ to $H_1$ must have central charge $-c$ mod $8$ assuming the fundamental degrees of freedom are bosonic, as total central charge must be $0$ mod $8$.
That is, the allowed central charge on the interface between any given pair of bulk Hamiltonians is some unique number mod $8$.  Given a bulk three-dimensional commuting projector Hamiltonian $H$, we now define
the $M$ invariant of $H$, written $M(H)$, to be the chiral central charge (mod $8$) of an interface from the trivial Hamiltonian to $H$.

Now consider three Hamiltonians $H_1,H_2,H_3$ with $H_1$ equal to the trivial Hamiltonian.  We consider a three-torus with three interfaces between three different bulk regions, taking $H_1,H_2,H_3$ for three bulk regions with coordinates
$0+\epsilon \leq \theta_3 \leq 2\pi/3-\epsilon$,
$2\pi/3+\epsilon \leq \theta_3 \leq 4\pi/3-\epsilon$, and
$4\pi/3+\epsilon \leq \theta_3 \leq 2\pi-\epsilon$, respectively.
Then, the interface from $H_1$ to $H_2$ has chiral central charge $M(H_2)$ mod $8$ while the interface from $H_2$ to $H_3$ has chiral central charge $-M(H_3)$ mod $8$.  Then, the interface from $H_2$ to $H_3$ must have chiral central charge
$M(H_3)-M(H_2)$ mod $8$, so indeed it suffices to know $M(H)$ for any given $H$ to be able to compute the chiral central charge of the
interface between two Hamiltonians.

\section{Topologically ordered boundary} \label{sec:topoboundary}

All dimensions will be spatial in this section.

When described by a quantum field theory, the boundary of a nontrivial SPT phase 
has an `t~Hooft anomaly for the global symmetry action~\cite{
BCFV, ProjS, TI, TopoSC, Vishwanath_2013, LevinGu, Kapustin, Kapustin_Thorngren1, Kapustin_Thorngren2}.
Thus it would be interesting to find such a boundary for our model.
Certainly we can simply naively truncate the Hamiltonian of our model in a symmetric way at the boundary,
but doing this results in a finely tuned surface state with exponential degeneracy, 
which is a situation not well described by a field theory.
In fact, by truncating the symmetric disentangler~$\Udis$, 
we can model the boundary as a stand-alone $3d$ lattice system, 
with symmetry acting as the nontrivial QCA $\Xbdry$.
The question then becomes: is there a $3d$ lattice Hamiltonian, invariant under $\Xbdry$, 
whose low energy spectrum is described by a simple field theory, 
and how does the symmetry act in that field theory?

In this section we will answer this question by describing the construction of a gapped boundary topological order symmetric under $\Xbdry$.  This topological order is simply a $\ZZ_2$ gauge theory with fermionic gauge charge.  We will refer to it as a $\ZZ_2^f$ gauge theory so as not to confuse the gauge and symmetry $\ZZ_2$'s.  Let us first show how to drive the boundary into this topological order, and then discuss the symmetry fractionalization.  We will only sketch the construction.  Note that in what follows, everything is occurring on the 3d boundary.  The bulk is not relevant to the discussion anymore; its only effect was to make the symmetry act as the nontrivial QCA $\Xbdry$.

First consider a general configuration of the spins.  
Let us put the ``material'' into a trivial product state in the up spin domains, 
and into the nontrivial 3-fermion Walker-Wang state in the down spin domains.
The Hamiltonian for the latter is just given by conjugating the trivial Hamiltonian by $\Xbdry$.
We would now like to form a superposition of such states over all spin configurations 
--- i.e. proliferate the domain walls --- to obtain a fully symmetric state.
The obstacle is that, with the naive commuting projector Hamiltonian for the material,
the domain walls host the 3-fermion topological order,
and hence a domain wall configuration of high genus has a large ground state degeneracy, 
preventing us from proliferating the domain walls.

One can attempt to solve this problem by introducing ancilla degrees of freedom, 
putting those ancillas in the 3-fermion topological order along the domain wall boundaries 
(and into a product state elsewhere), 
and then condensing the appropriate bound states to make the domain walls topologically trivial.
However, the Hamiltonian for the ancillas is effectively $2$-dimensional, 
and hence gives chirality $c_{-} = \pm 4$ to the domain walls.
One then has to make a sign choice for the chirality, or, more precisely, 
choose a normal direction to the domain wall, everywhere along the domain wall.
Using an argument similar to that in the footnote in \cref{subsec:l}, 
we see that the QCA $\Xbdry$ cannot reverse this chirality, 
so we cannot e.g. do the naive thing and orient the normal from say the spin up region to the spin down region;
if we did, then the state would not be symmetric under $\Xbdry$, 
which flips the spins but not the chirality.
We must instead pick normal directions using only the information about the domain wall locations; 
however, this is impossible to do in a consistent way.
If we pick the normal directions in an arbitrary (i.e. inconsistent) way, 
we will end up with $1$-dimensional edges on the domain walls 
where the normal flips; 
these edges will host gapless modes (equivalent to an odd number of $E_8$ edges)
that again prevent us from proliferating the domain walls and building a gapped symmetric state.

It is not surprising that our naive attempt at building a symmetric state failed; 
had it succeeded, we would have ended up with a short range entangled gapped symmetric boundary 
for a supposedly nontrivial SPT, which is not possible.  
However, a variation on the above attempt does succeed, 
at the expense of ending up with a surface topological order.
The variation involves introducing ancilla degrees of freedom at the boundary.
These are bosonic ancillas, but we will think of them as fermions $f$ coupled to a lattice $\ZZ_2^f$ gauge field,
similar to the way the spin degrees of freedom in Kitaev's honeycomb model 
can be thought of as fermions coupled to a $\ZZ_2^f$ gauge field~\cite{Kitaev_2005, Kapustin2d, Kapustin3d,verstraete2005mapping}.  
We now proceed as before, but make the domain walls topologically trivial 
by condensing a bilinear between $f$ and one of the $3$-fermion anyons 
--- such a particle is a boson, and can be condensed everywhere on the domain wall.
This can presumably be done entirely within the space of commuting projector models,
although we have not worked out the details.
After the domain walls are proliferated, we end up with a symmetric state $\ket{\Psi_{\rm{bdry}} }$ 
that is a finite depth circuit away from a trivial product state tensored with a fermionic $\ZZ_2^f$ gauge theory.
Hence $|\Psi_{\rm{bdry}}\rangle$ is itself a fermionic $\ZZ_2^f$ gauge theory.

One piece of intuition for why introducing fermions allows us 
to solve the chirality problem and proliferate the domain walls 
is the fact that the minimum quantized value of 
chiral central charge for $2d$ short range entangled fermionic systems is $\frac{1}{2}$ 
(realized by a $p+ip$ superconductor), rather than $8$.
Thus we could have proceeded as in the naive approach by arbitrarily picking the normal 
that gaps out the domain walls, and then using the fermions, 
in the appropriate multiple of $p+ip$ states ($8$ or $-8$), 
to cancel the chirality that obstructed that approach.
This ends up being equivalent to what was described above.

Let us now describe what makes this $\ZZ_2^f$ topological order anomalous under the global $\ZZ_2$ symmetry.
This is best discussed in the framework of Chen and Hermele~\cite{HermeleChen},
who define symmetry fractionalization for $3d$ symmetry enriched topological (SET) orders as follows: 
first, dimensionally reduce the system along one direction (say $z$), 
so the geometry is ${\mathbb{R}}^2 \times S^1$.  
Then put in a symmetry defect, which one thinks of as a large loop along the $xy$-plane (at a fixed $z$),
separating this $xy$-plane into an inside and an outside.
Finally, examine the quasi-2d SET order both inside and outside the loop.
The difference between these SET orders gives information about the symmetry fractionalization.
In all the examples considered in \cite{HermeleChen}, 
the two quasi-$2d$ topological orders were the same (namely the toric code),
and only their symmetry properties differed.
Applying this construction to our example, 
however, we end up with different quasi-$2d$ topological orders.

To see this, note that without the symmetry defect, 
the quasi-$2d$ topological order is just the toric code.
The quasi-$2d$ pointlike excitations are the fermion $f$,
a short $\ZZ_2^f$ flux loop wound along the $z$ direction (a boson), and their bound state.
However, inside the symmetry defect loop, 
the short $\ZZ_2^f$ flux loop crosses the $\ZZ_2$ domain wall and changes statistics to become fermionic.
Indeed, because of the condensation on the domain wall, 
one of the $3$-fermion anyons becomes bound to this short $\ZZ_2^f$ flux loop.
(If $f f_1$ is the condensed boson on the domain wall, the short $\ZZ_2^f$ flux loop is bound to $f_2$.)
Another way to see that the loop becomes fermionic 
is from the perspective of cancelling chirality: 
we put the $f$ fermions into $8$ copies of a $p+ip$ state along the domain wall,
which changes the statistics for corresponding $\pi$ flux.
The bound state of the $\ZZ_2^f$ flux loop and the fundamental fermion is also fermionic 
(indeed, once we have $2$ fermions in a $2d$ topological order with $\ZZ_2 \times \ZZ_2$ fusion rules, 
we must have three).
Thus inside the defect loop we have the $3$-fermion topological order, 
in contrast to the toric code outside.

The argument that such an SET is anomalous is as follows.  Supposing for a contradiction that one could realize such an SET purely in 3D, it would follow, using the above dimensional reduction thought experiment and the relation between 2d anyon statistics and chiral central charge \cite{Kitaev_2005} that a symmetry defect loop carries non-zero chirality.  In particular, fusing two such identical symmetry defect loops would give a state which has non-zero chirality.  However, this state can simply be created from the ground state by a unitary acting only in the vicinity of a 1d loop, which is impossible.  See \cite{CTW} for more details (in particular their condition (8i), discussed in their Appendix B).

In general we expect that the difference in chiral central charges mod $8$,
computed from the statistics of the two topological orders, 
should be an anomaly indicator; 
namely, the anomaly is nontrivial when this chiral central charge difference is $4$ mod $8$.

\section{Miscellaneous Discussion}
In this section we speculatively discuss miscellaneous aspects of the model.

\subsection{Consequences of a Circuit Disentangling the Three-fermion Ground State}

First, we consider the consequences if some hypothetical quantum circuit exists 
to disentangle the ground state of the three-fermion Walker-Wang model.%
\footnote{
Note that we implicitly assumed the existence of such a circuit 
in the definition of the bulk invariant in \cref{sec:bulkinvariant}.
}
We emphasize that this circuit would not map the separator~\cite{haah2018nontrivial} of that model 
to the trivial separator as $\alpha_{WW}$ does (we may refer to this as ``disentangling the Hamiltonian''),
but would only disentangle the ground state.
Although we do not have such a circuit, 
there also does not seem to be any obvious obstruction to the existence of such a circuit,
though it may need to have tails; 
i.e., rather than the gates being strictly local, 
perhaps the gates need to be well-approximated by strictly operators, up to some decaying tail.

We expect that such a circuit, if it is truncated to some finite region, 
will create a chiral state on the two-dimensional boundary of that region with $c=\pm 4$.

Suppose such a circuit did exist, implementing some unitary $U$; 
more precisely, suppose that a set of such unitaries existed, 
one for each three-manifold $M$, calling the unitary on the given manifold $U(M)$.
One might try then to use this set to define the decorated domain wall state $\Psi_0$, 
rather than using our construction involving $\alpha_{WW}$.  
That is, for each spin configuration, $\vec z$, 
one might create a three-fermion Walker-Wang state on the boundary using $U(M)$.

However, a problem arises.  If we did not care about preserving the $\mathbb{Z}_2$ symmetry of the model, we could orient the domain walls from up spins to down spins.  However, we want to preserve this symmetry; this is the reason that we chose previously some fixed projection on each cell, independent of spin configuration.
In this case, we might need to ``stitch" together different circuits on different balls.  In some cases, we may expect that the chirality will cancel, giving a state with vanishing chirality on the two-dimensional boundary, but in other cases the chirality may add, giving a state with chirality $\pm 8$ on the boundary.  This means that we may create some fluctuating configuration of $E_8$ states with chirality $\pm 8$.  We do not expect that such a state can be disentangled by a circuit that preserves the $\ZZ_2$ symmetry at the level of gates.  So, we do not expect that there is a circuit that respects the ${\mathbb Z}_2$ symmetry at the level of gates and fully disentangles $\Psi_0$.  While this argument is heuristic, it gives some reason to believe that $\alpha_{WW}$ was essential to our construction.

Another interesting consequence of the existence of such a circuit implementing a unitary $U$ is that if we conjugate the three-fermion Walker-Wang Hamiltonian by $U$, the resulting Hamiltonian will be a commuting projector Hamiltonian whose ground state is a trivial product state.  At the same time, we expect that (from the $M$ invariant) a domain wall between this Hamiltonian and the obvious Hamiltonian which stabilizes the trivial ground state (i.e., the sum of Pauli $Z$ operators on each qubit) will carry a chiral central charge $4$ mod $8$.

Hence this leads to a natural conjecture that the space of local commuting projector Hamiltonians whose ground state is a trivial product state is not a connected space.  A precise formulation of this conjecture would require specifying what we mean by a local Hamiltonian since the disentangling circuit may have tails.

\subsection{Comparison to Two-Dimensional Beyond Cohomology Phases}
Another interesting phase to compare to is a two-dimensional decorated domain wall state with fluctuating spins and domain walls decorated by Majorana chains \cite{Tarantino, Ware} in a nontrivial state.  That is, one has two Majorana fermions per bond.  On bonds which are not in a domain wall, the pair of Majoranas $\gamma,\gamma'$ are in an eigenstate $\gamma \gamma'=\pm i$, but on domain walls with bonds labelled $1,2,3,\ldots$ and Majorana operators $\gamma_j,\gamma'_j$ one instead has $\gamma'_j \gamma_{j+1}=\pm i$, taking the index $j$ periodic in the obvious way.

This state likely can be disentangled by a quantum circuit since it is a two-dimensional state without anyons and with strictly finite-range correlations.  However, it seems that there may be no circuit which disentangles a ${\mathbb Z}_2$ symmetric Hamiltonian for this state.  Equivalently, it may not be possible to find a quantum circuit which disentangles the state and which obeys the ${\mathbb Z}_2$ symmetry (note that if such a circuit existed, one could conjugate a trivial Hamiltonian by this circuit to get a Hamiltonian whose ground state is the decorated domain wall state).

One might think that one could construct such a circuit using similar ideas to what we have done here.  The Majorana chain can be constructed on a domain wall by applying a QCA which shifts Majorana operators by $1$ along the domain wall, i.e., maps $\gamma_j$ to $\gamma'_j$ and maps $\gamma'_j$ to $\gamma_{j+1}$.
However, if this procedure did work, one would find (similarly to what we found in our four-dimensional model) that the boundary operator $\Xbdry$ would be equivalent to such a shift up to a quantum circuit.  Then, the square of the boundary operator would be a shift by $2$, again up to a quantum circuit.  Since a shift by two is a nontrivial QCA in one-dimension, the boundary operator cannot square to the identity, giving a contradiction.

The problem with applying our construction to this two-dimensional system can also be phrased differently: in some cases we will shift by one to the left on a domain wall and in some cases we will shift by one to the right and we cannot interpolate between these two shifts in a one-dimensional QCA.  If we did not require that the circuit respect the ${\mathbb Z}_2$ symmetry there is no such problem: we can always decide to shift clockwise (or always shift counter-clockwise) around any spin up domain and this gives us a local rule to decide the direction of the shift, using the additional information about the spin configuration.

\subsection{Connection to anomalies in $3+1$ dimensions}

Although models in spatial dimensions higher than $3$ are not directly relevant to our $3+1$ dimensional world, the `t~Hooft anomalies that they realize at their boundaries might be.  In \cref{sec:topoboundary} we discussed one particular symmetric topologically ordered boundary termination of our model that saturates the $\ZZ_2$ boundary anomaly.  It would be interesting to extend this work to more general systems and symmetries.  For example, there exist $U(1)$-protected beyond cohomology SPT phases of both bosons and fermions which are natural generalizations of the $\ZZ_2$ beyond cohomology SPT studied in this paper \cite{Wang_2015}.  Can a quantum-information theoretic QCA framework be useful in understanding the boundary action of this $U(1)$?  This would be especially interesting for the case of fermions, where the SPTs in question are just the $4+1$ dimensional integer quantum Hall phases.  The boundary action of the $U(1)$ symmetry in that case realizes the chiral anomaly, and is saturated by a boundary state consisting of some number of Weyl fermions of the same chirality.  It would be very interesting if the QCA framework had some relevance to understanding such perturbative chiral anomalies.

Although for a continuous symmetry, the boundary symmetry operator is necessarily connected to the identity by a continuous path (and hence cannot be a nontrivial QCA), it is possible that the path of boundary symmetry operators rotating phase by a full $2 \pi$ is a nontrivial path, giving rise to a nontrivial QCA in one lower dimension.  For a three-dimensional boundary, this requires a nontrivial QCA in two dimensions, which has been ruled out for systems with only bosonic degrees of freedom~\cite{FreedmanHastings} but remains open in the fermionic case.

\subsection{QCA boundary action and symmetry group}

Recall that the boundary symmetry action following \cite{ElseNayak}
is a representation of the symmetry group $G$.
In our construction this representation resulted 
in a nontrivial group $Q$ of QCA modulo quantum circuits.
Abstractly, this means that we have a group homomorphism~$G\to Q$
which has to factor through the abelianization $G/[G,G]$ as
\begin{align}
G \to G/[G,G] \to Q
\end{align}
since $Q$ is always an abelian group.
Hence, in order for the $J$ invariant to be nontrivial,
at least, $G/[G,G]$ has to be nontrivial.
For example, if $G$ is a nonabelian simple group such as $A_5$ (the alternating group of order~$60$)
then $G/[G,G] = 0$,
and no SPT under such a group can induce nontrivial QCA on the boundary.

\subsection{Extension to Other Beyond Cohomology SPT Phases}
It is interesting to ask whether this kind of construction can be extended to other beyond cohomology SPT phases.  A key role was played in our construction by the QCA $\alpha_{WW}$ which disentangles the three-fermion Walker-Wang model.  We expect that other QCA which are nontrivial (and further, which are nontrivial modulo shift QCA) could be used to construct beyond cohomology phases also, although of course many details of the construction would have to be checked in the case of some other QCA.

In fact, we do have other candidates for 
QCA which are nontrivial modulo shifts in three spatial dimensions\cite{haah2019clifford}, though we do not have any examples outside three dimensions yet.  Like $\alpha_{WW}$, these other QCA also are (generalized) Clifford QCA: although they act on qudits of prime dimension $p>2$, rather than qubits, they map (generalized) Pauli matrices to products of Pauli matrices.

Two different cases are observed for these generalized Clifford QCA.
In one case, similar to $\alpha_{WW}$, the square of the QCA is trivial (up to shift), 
and we expect that one could construct $\ZZ_2$ beyond cohomology SPTs with these QCA.
In the other case, the square is nontrivial but the fourth power is trivial
and we expect that one could construct $\ZZ_4$ beyond cohomology SPTs with these QCA.
Correspondingly, these QCA disentangle a three-dimensional bulk whose boundary has chiral central charge 
$4\,\,\rm{mod}\,\,8$ or $\pm 2\,\,\rm{mod}\,\,8$.
We emphasize that this means that the spins and material have different Hilbert space dimensions: 
the spins have dimensions $2$ or $4$ while the material has dimension $p$ for prime $p>2$.

A final interesting question is: 
if we consider just the constructions using $\ZZ_2$ symmetry for the spins and using QCA 
whose square is trivial, 
do different choices of QCA correspond to different beyond cohomology SPT phases?
This is likely equivalent to the question:
given two three-dimensional QCA $\alpha$ and $\beta$ which each disentangle a three-dimensional bulk 
whose boundary has chiral central charge $4\,\,\rm{mod}\,\,8$, are $\alpha$ and $\beta$ equivalent up to circuits?

As an example, suppose that $\alpha$ disentangles the three-fermion Walker-Wang model, 
while $\beta$ is a QCA which disentangles a bulk whose boundary is a $\ZZ/5$ dyon model~\cite{haah2019clifford}.
(Here equivalence up to circuits should be a notion of stable equivalence where one can tensor in additional degrees of freedom on which $\alpha$ or $\beta$ acts trivially; such a stable notion is necessary since $\alpha$ acts on qubits and $\beta$ acts on qudits of dimension $5$.)
This $\ZZ/5$ dyon theory has five anyons obeying the fusion rule of the additive group $\ZZ/5$
whose topological spins are $e^{4\pi i k^2 / 5}$ for $k = 0,1,2,3,4$.
The answer to the question of the equivalence of $\alpha$ and $\beta$ is almost certainly equivalent to the question: 
is there a two-dimensional commuting projector Hamiltonian which describes the three-fermion TQFT 
tensored with this $\ZZ/5$ dyon TQFT?  
This model has chiral central charge $0\,\,\rm{mod}\,\,8$ so there is no obvious obstruction.
However, there is a conjecture~\cite{kevin} that every commuting projector model in two-dimensions 
describes some quantum double, and this particular TQFT is {\it not} a quantum double.
Thus, if this conjecture holds, then $\alpha$ and $\beta$ are not equivalent up to circuits.

It is conceivable that the resolution to this puzzle (whether such $\alpha$ and $\beta$ are equivalent) 
depends upon what kind of ``equivalence under circuits'' is allowed.
If one allows only bounded depth circuits consisting of strictly local gates,
then we expect that $\alpha$ and $\beta$ are not equivalent.
On the other hand, if one allows a broader notion of circuit equivalence,
such as allowing a ``circuit'' to be the unitary given by evolution for bounded time under a time-dependent Hamiltonian whose terms are approximately local (for example, having exponential tails), 
then one may expect that $\alpha$ and $\beta$ are equivalent.

As an analogy to this way in which different notions of locality 
can lead to different notions of equivalence, 
consider the question of classifying manifolds in topology.
There, various categories are considered, such as DIFF, PL, and TOP, 
corresponding to different requirements on the smoothness of the manifold.
The classification of manifolds in general depends upon the particular category chosen.
So, we propose that one may define a ``strictly local'' category 
whose objects are QCA obeying a strict notion of locality, 
with morphisms which are bounded depth quantum circuits with strictly local gates.
One may also define a ``tails'' category, 
whose objects are again QCA but with a more relaxed notion of locality 
and with similarly more relaxed morphisms.
We leave it to future work to determine an appropriate notion of locality for the tails category 
(for example, exponential or power law); 
note that even in one-dimension, 
the question of classifying QCA with an approximate notion of locality is completely open,
though the question with strict locality is completely solved in \cite{Gross_2012}.

\begin{acknowledgments}
We thank D. Freed, A. Vishwanath, C. Xu, and Z. Wang for useful discussions.  LF and MH thank the
 Kavli Institute for Theoretical Physics for hospitality, supported by the National Science Foundation under Grant No. NSF PHY-1748958.
\end{acknowledgments}

\bibliography{4Dbeyond-ref}
\end{document}